\documentclass[runningheads]{llncs}

\usepackage[T1]{fontenc}
\usepackage{xcolor}
\usepackage{graphicx}
\usepackage{colortbl}
\usepackage{subcaption}
\usepackage{algorithm,algpseudocode}
\usepackage{bbding}
\usepackage[disable]{todonotes}
\usepackage{xspace}
\usepackage{amssymb,mathtools}
\usepackage[hidelinks]{hyperref}
\usepackage{orcidlink}

\newcommand{\ourprob}{{\sc MCSFG}}

\newcommand{\degree}{\mathrm{degree}}

\newcommand{\pw}{\mathrm{pw}}
\newcommand{\td}{\mathrm{td}}
\newcommand{\tw}{\mathrm{tw}}
\newcommand{\br}{\mathsf{br}}
\newcommand{\SOL}{\textsf{SOL}}
\newcommand{\cc}{\mathrm{cc}}

\newcommand{\mces}{{\sc Maximum Common Edge Subgraph}\xspace}
\newcommand{\mcis}{{\sc Maximum Common Induced Subgraph}\xspace}
\newcommand{\siso}{{\sc Subgraph Isomorphism}\xspace}
\newcommand{\mcics}{{\sc Maximum Common Induced Connected Subgraph}\xspace}
\renewcommand{\AA}{\ensuremath{\mathcal A}\xspace}
\newcommand{\HH}{\ensuremath{\mathcal H}\xspace}
\newcommand{\OO}{\ensuremath{\mathcal O}\xspace}
\newcommand{\NB}{\ensuremath{\mathbb N}\xspace}
\newcommand{\el}{\ensuremath{\ell}\xspace}
\newcommand{\WOH}{\ensuremath{\mathsf{W[1]}}-hard\xspace}
\newcommand{\WTH}{\ensuremath{\mathsf{W[2]}}-hard\xspace}
\newcommand{\FPT}{\ensuremath{\mathsf{FPT}}\xspace}
\newcommand{\NPH}{\ensuremath{\mathsf{NP}}-hard\xspace}
\newcommand{\NPC}{\ensuremath{\mathsf{NP}}-complete\xspace}
\newcommand{\PNPH}{para-\ensuremath{\mathsf{NP}}-hard\xspace}
\newcommand{\NP}{\ensuremath{\mathsf{NP}}\xspace}
\newcommand{\yes}{{\sc yes}\xspace}
\newcommand{\no}{{\sc no}\xspace}
\newcommand{\true}{\text{{\sc true}}\xspace}
\newcommand{\false}{\text{{\sc false}}\xspace}
\newcommand{\OPT}{\ensuremath{\mathsf{OPT}}\xspace}
\newcommand{\Pb}{\ensuremath{\mathsf{P}}\xspace}
\newcommand{\PP}{\ensuremath{\mathcal P}\xspace}
\newcommand{\TT}{\ensuremath{\mathcal T}\xspace}
\DeclareMathOperator*{\argmin}{arg\,min}
\DeclareMathOperator*{\argmax}{arg\,max}
\newtheorem{observation}{\bf Observation}

\spnewtheorem{claim}{Claim}{\itshape}{\rmfamily}
\newcommand{\defproblem}[3]{
  \vspace{1mm}
  \begin{center}
    \noindent\fbox{
      \begin{minipage}{0.95\textwidth}
        \textsc{\underline{#1}}\\[1ex]
        \textbf{Input:} #2\\
        \textbf{Question:} #3
      \end{minipage}
    }
  \end{center}
  \vspace{1mm}
}
 \renewcommand{\orcidID}[1]{\unskip\,\orcidlink{#1}}

\begin{document}

\title{Parameterized Complexity of Finding a Maximum Common Vertex Subgraph Without Isolated Vertices}

\titlerunning{Maximum Common Vertex Subgraph Without Isolated Vertices}

\author{Palash Dey\inst{1}\orcidID{0000-0003-0071-9464} \and
Anubhav Dhar\inst{2,3}\orcidID{0009-0006-5922-8300} \and
Ashlesha Hota\inst{1}\orcidID{0009-0009-8805-4583} \and
Sudeshna Kolay\inst{1}\orcidID{0000-0002-2975-4856} \and
Aritra Mitra\inst{4}\thanks{Most of the work was performed when the author was at Indian Institute of Technology Kharagpur.} \orcidID{0000-0003-0248-1509} \Envelope
}

\authorrunning{P. Dey, A. Dhar, A. Hota, S. Kolay, A. Mitra}

\institute{Indian Institute of Technology Kharagpur, India \email{palash.dey@cse.iitkgp.ac.in, ashleshahota.23@kgpian.iitkgp.ac.in, skolay@cse.iitkgp.ac.in}
\and
Max Planck Institute for Informatics, Saarbr\"{u}cken, Germany
\and
Saarland University, Germany \email{adhar@mpi-sws.org}
\and
CISPA Helmholtz Center for Information Security \email{\Envelope aritra.mitra@cispa.de}
}

\maketitle

\begin{abstract}

In this paper, we study the {\sc Maximum Common Vertex Subgraph} problem: Given two input graphs $G_1$, $G_2$, and a non-negative integer $h$, is there a common subgraph $H$ on at least $h$ vertices such that $H$ has no isolated vertices? In other words, each connected component of $H$ has at least two vertices. This problem naturally arises in graph theory along with other variants of the well-studied {\sc Maximum Common Subgraph} problem and also has applications in computational social choice.

This problem is equivalent to the problem of {\sc Maximum Common Star Forest Subgraph (MCSFG)}, which asks to decide, for graphs $G_1$, $G_2$ and an integer $h$, whether there is a star forest $H$ on at least $h$ vertices that is a common subgraph of both $G_1$ and $G_2$. We show that this problem is \textsf{NP}-hard and provide an \textsf{FPT} algorithm when parameterized by $h$.

Next, we study the problem with respect to common structural parameters such as vertex cover number, maximum degree, treedepth, pathwidth, and treewidth of one or both input graphs. We derive a complete dichotomy of parameterized results with respect to individual parameterizations as well as combinations of parameterizations involving the above structural parameters. This provides insight into the complexity-theoretic and parameterized landscape of the problem.


\end{abstract} 
\section{Introduction}

The problem of finding a \emph{maximum common subgraph} of two given graphs is well-studied in algorithmic and structural graph theory. The most prominent variants include (i) \mcis, asking to maximize the number of vertices of a common induced subgraph, (ii) \mces, asking to maximize the number of edges of a common subgraph, and (iii) \mcics, asking to maximize the number of vertices in a common connected induced subgraph. A similar natural problem is to ask for a common subgraph with the maximum number of vertices for two input graphs $G_1$ and $G_2$. However, this problem admits a trivial solution: a graph consisting of $\min(|V(G_1)|,|V(G_2)|)$-many isolated vertices. Evidently, such a trivial graph maximizes the number of vertices among all common subgraphs. 

In this paper, we consider the next natural variant, the problem of {\sc Maximum Common Vertex Subgraph}, which asks to maximize the number of vertices in a common subgraph of two input graphs, but where the common subgraph has no isolated vertices. Interestingly, forbidding isolated vertices makes this problem significantly non-trivial, and in this paper we show that {\sc Maximum Common Vertex Subgraph} is \NPH and boasts a rich landscape in parameterized complexity. We study the parameterized complexity of {\sc Maximum Common Vertex Subgraph} under various choices of parameters: the size of the maximum common subgraph, maximum degree, treewidth, pathwidth, treedepth, and vertex cover number.

Apart from being an interesting and natural problem in graph theory in its own right, {\sc Maximum Common Vertex Subgraph} finds its application in graphical variants of the house allocation problem in computational social choice. We state the following motivation. Let $\AA$ be a set of agents and $\HH$ be a set of houses. Let $G_\AA$ be a graph with $\AA$ as the vertex set, where edges depict some social relationship among the agents (e.g. friendship, acquaintances, collaborators, etc). Similarly, let $G_\HH$ be a graph with $\HH$ as the vertex set, where edges depict some notion of `nearness' (e.g. geographic, transit-accessible, etc). Each agent $a \in \AA$ has to be allocated a unique house $\phi(a) \in \HH$. An agent $a \in \AA$ is \emph{happy} if at least one of the neighbours of $a$ (in $G_\AA$) is allocated to a neighbouring house (in $G_\HH$) of $\phi(a)$. The problem of {\sc Maximum Neighbour-Happy House Allocation} asks to find an allocation which maximizes the number of happy agents. Although being a natural extension to other well-studied variants of graphical house allocation~\cite{hosseini2024graphical,hosseini2023tight,hosseini2023graphical}, the problem is identical to the {\sc Maximum Common Vertex Subgraph} problem on the graphs $G_\AA$ and $G_\HH$.

We formally state the problem of {\sc Maximum Common Vertex Subgraph}.

\defproblem{Maximum Common Vertex Subgraph}{Undirected unweighted graphs $G_1$, $G_2$, non-negative integer $h$.}{Does there exist a graph $H$ without isolated vertices with $|V(H)| \ge h$, such that $H$ is a subgraph of both $G_1$ and $G_2$?}

We provide further structural insight to the maximum common vertex subgraph $H$ of $G_1$ and $G_2$. A spanning forest of $H$ is also a common subgraph without isolated vertices of $G_1$ and $G_2$ with the same number of vertices as $H$. This allows us to restrict to only common forest subgraphs of $G_1$ and $G_2$. Taking this one step further, any subgraph induced by a \emph{minimum edge cover} of $H$ would also be a common subgraph of $G_1$ and $G_2$ without isolated vertices. Recall that an \emph{edge cover} of $H$ is a subset $S$ of its edges with every vertex of $H$ lying on at least one edge of $S$. A \emph{minimum edge cover} is an edge cover of minimum size. Recall that a \emph{non-trivial star} is a tree of at least two vertices with diameter at most $2$. A minimum edge cover of a graph without isolated vertices is known to induce a collection of non-trivial stars~\cite{gallai1959uber,west2001introduction}. Let a \emph{star forest} denote a disjoint union of non-trivial stars. This observation lets us restrict our search to only star forests appearing as common subgraphs. The problem {\sc Maximum Common Star Forest Subgraph (\ourprob)} is stated formally as follows, and since it is identical to the problem of {\sc Maximum Common Vertex Subgraph}, we restrict our focus throughout the paper on \ourprob{}.

\defproblem{Maximum Common Star Forest Subgraph (\ourprob)}{Undirected unweighted graphs $G_1$, $G_2$, non-negative integer $h$.}{Does there exist a star forest $H$ on at least $h$ vertices, such that $H$ is a subgraph of both $G_1$ and $G_2$?}

\subsection{Related Work}
The \siso problem asks whether there exists an injective mapping $\phi: V(H) \to V(G)$ from a pattern graph $H$ to a host graph $G$ such that $\{u,v\} \in E(H)$ implies $\{\phi(u),\phi(v)\} \in E(G)$. This problem is \NPH~\cite{garey2002computers}. In fact a trivial reduction from {\sc Maximum Clique} shows that the problem is \WOH when the size of the pattern graph is the parameter.
Kann~\cite{kann1992approximability} observes that the problems of \mcis, \mces and \mcics generalize the \siso problem, and hence are \NPH. Their work further studies the approximability of these problems.
\mcis is shown to be at least as hard to approximate as {\sc Maximum Clique}. The \mcics is $\mathsf{NPO}$ $\mathsf{PB}$-complete. Kann also related node- and edge-based formulations via L-reductions, establishing foundational approximation and hardness results.
Khanna et al.~\cite{khanna1995approximation} studied the \textsc{Largest Common Subtree} problem, which asks for the largest tree that appears as a common subgraph in a collection of input trees. They showed that for bounded-degree trees, an approximation ratio of $\OO\!\left(\frac{n \log \log n}{\log^2 n}\right)$ is achievable in polynomial time. For unbounded-degree trees, an $\OO\!\left(\frac{n (\log \log n)^2}{\log^2 n}\right)$-polynomial time approximation algorithm is known in the unlabelled case. The same approximation guarantee holds for labelled trees of unbounded degree when the number of distinct labels is $\OO(\log^{\OO(1)} n)$. Furthermore, Akutsu et al.~\cite{akutsu2015complexity} studied the problem for a restricted unordered setting where the maximum outdegree of the common subtree is bounded by a constant. 
Yamaguchi et al.~\cite{yamaguchi2004finding} further imposed structural restrictions, focusing on tractable cases, in particular, graphs with bounded degree, treewidth, and a polynomial number of spanning trees. Exploiting these properties, a polynomial-time algorithm was developed for computing a maximum common connected induced subgraph between a degree-bounded partial $k$-tree and a connected graph whose number of spanning trees is polynomially bounded.

From the parameterized complexity perspective, \textsc{Maximum Common Induced Subgraph} (MCIS) is \WOH when parameterized by treewidth. This follows from the NP-hardness of \textsc{Subgraph Isomorphism} on graphs of bounded treewidth~\cite{abu2014maximum,akutsu2012complexity,matouvsek1992complexity}. 

Abu-Khzam~\cite{abu2014maximum} further showed that MCIS, when parameterized by the size of a minimum vertex cover of only one of the input graphs, remains \WOH even when restricted to bipartite graphs. In contrast, when parameterized by the sum of the vertex cover numbers of both input graphs, MCIS is fixed-parameter tractable (\FPT).

Gima et al.~\cite{gima2022exploring} introduced the graph parameter \emph{vertex integrity} ($vi$) to bridge the gap between treedepth and vertex cover. The vertex integrity of a graph $G$, denoted by $vi(G)$, is defined as $vi(G) = \min_{S \subseteq V(G)} \big( |S| + \max_{C \in \cc(G - S)} |V(C)| \big)$, where $\cc(G - S)$ denotes the set of connected components of $G - S \colon=  G[V(G) \setminus S]$. Using this notion, they generalized several \FPT results known for vertex cover to vertex integrity. In particular, they showed that both \textsc{Maximum Common Edge subgraph} (MCES) and MCIS are \FPT when parameterized by $vi(G_1) + vi(G_2)$.

Hanaka et al.~\cite{hanaka2026finding} studied the parameterized complexity of MCES and MCIS with respect to max-leaf number and neighbourhood diversity. The \emph{max-leaf number} of a connected graph $G$, denoted by $ml(G)$, is the maximum number of leaves in any spanning tree of $G$, where a leaf is a vertex of degree one in the tree. For a disconnected graph $G$, the max-leaf number is defined as $ml(G) = \sum_{C \in \cc(G)} ml(C)$, where $\cc(G)$ denotes the set of connected components of $G$. From the definition, it follows that $ml(G) \geq \Delta(G)$, where $\Delta(G)$ is the maximum degree of $G$. The \emph{neighbourhood diversity} of a graph $G$, denoted by $nd(G)$, is the number of twin classes in $G$, where each twin class induces either a clique or an independent set, and between any two twin classes, either all possible edges are present or none are present. Using these properties, they showed that both MCES and MCIS are \FPT when parameterized by $ml(G_1) + ml(G_2)$, and also when parameterized by $nd(G_1) + nd(G_2)$.

On the structural side, Kriege et al.~\cite{kriege2018maximum} studied the \textsc{Maximum Common Connected Subgraph} problem, where given two graphs $G$ and $H$, the goal is to return the number of vertices in a maximum common connected induced subgraph. They proved that this remains \NPH even for biconnected series-parallel graphs (partial 2-trees) where all but one vertex have degree at most three.

Hoffmann et al.~\cite{hoffmann2017between} introduced $k$-\textsc{less subgraph isomorphism}, where up to $k$ vertices of the pattern graph may be deleted before embedding. They developed weakened filtering invariants generalizing degree- and path-based constraints, implemented in a constraint-programming style algorithm.

As we showed, our problem can be equivalently thought of as finding the largest common star forest. Star-packing in graphs is very well-studied. Prieto and Sloper~\cite{prieto2006looking} studied packing $k$ vertex-disjoint copies of a graph $H$ into $G$, focusing on the case where $H$ is a star, $K_{1,s}$. For the special case of $s = 2$, i.e., $H = K_{1,2}$, they provided a linear kernel and an algorithm running in time $\OO(2^{5.301k} \cdot k^{2.5} + n^3)$, improving upon the quadratic kernel for the general packing of $k$ copies of $H = K_{1,s}$. Xi and Lin~\cite{xi2021maximum} investigated maximum $P_3$ packing problem on claw-free cubic planar graphs, showing it to be \NPH even under this structural restriction. In their follow up work~\cite{xi2024maximum}, they proved \NP-hardness for $K_{1,3}$-packing on claw-free cubic graphs. They designed a linear-time approximation algorithm covering at least $(3n-8)/4$ vertices, leveraging structural properties and local augmentation arguments. Huang et al.~\cite{huang2025approximation} studied $k^+$-star packing (i.e., packing of stars with at least $k$ leaves) and improved approximation guarantees, achieving an approximation ratio of $9/5$ for $k = 2$ and $(k+2)/2$ for general $k \geq 2$ using a local search framework. 

\subsection{Our Contributions}

We show that \ourprob{} is an \NPH problem. In fact, we show that \ourprob{} is \NPH even when the input graphs are planar with maximum degree at most $3$ (Theorem~\ref{thm:3-planar-nph}). In the rest of the paper we provide the parameterized complexity landscape of \ourprob{}. We start by studying the problem parameterized by the size of the common star forest. Theorem~\ref{thm:FPT-objective} shows that \ourprob{} admits an \FPT algorithm in that case. This is in stark contrast with the problems \mcis, \mces and \mcics being \WOH as they generalize {\sc Maximum Clique}.

Next, we study \ourprob{} parameterized by various structural parameters of the input graphs. In particular we consider the parameters of vertex cover number, treedepth, pathwidth and treewidth. Such structural parameters are frequently bounded in graphs that appear in real applications, and are exploited to push beyond the limits of otherwise intractable problems~\cite{cygan2015parameterized}. For a graph $G$, it is well known~\cite{gima2022exploring} that $\tw(G) \le \pw(G) \le \td(G) - 1\le \tau(G)$ where $G$ has treewidth $\tw(G)$, pathwidth $\pw(G)$, treedepth $\td(G)$ and vertex cover number $\tau(G)$. Alongside these parameters, we also consider the maximum degree $\Delta(G)$ of the graph, and explore all possible combinations of the parameter $\Delta(G)$ with the structural parameters $\tw(G)$, $\pw(G)$, $\td(G)$, and $\tau(G)$. 

First, when the parameter is the vertex cover number of one of the graphs, then \ourprob{} is \WTH (Theorem~\ref{thm:wth-vc}). In contrast, we show in Theorem~\ref{thm:FPT-vc} that \ourprob{} admits an \FPT algorithm when parameterized by the maximum vertex cover number of both input graphs.

This motivates us to study the parameterized complexity of \ourprob{} with respect to other structural parameters of both input graphs. We use a result due to Bodlaender et al.~\cite{bodlaender2020subgraph} to show that \ourprob{} is \PNPH when parameterized by the treedepth of both input graphs (Theorem~\ref{thm:td-paraNP}). This directly implies \PNPH{ness} when parameterized by pathwidth or treewidth of both input graphs. 

The \PNPH{ness} results of \ourprob{} when maximum degree of both graphs is the parameter, and when structural parameters of both graphs is the parameter motivate us to study the parameterized complexity when we look into a combination of these parameters. Indeed, we show that \ourprob{} admits an \FPT algorithm when parameterized by the combined parameter of treedepth and maximum degree of both input graphs (Corollary~\ref{cor:FPT-td-maxd}). We show this by first providing an \FPT algorithm solving \ourprob{} parameterized by the maximum size of the connected components across both graphs (Theorem~\ref{thm:FPT-conn-comp}).

Next, we consider parameterization by pathwidth and maximum degree of both input graphs and show that \ourprob{} is \WOH parameterized by the maximum degree of both input graphs on forests of pathwidth at most 4 (Theorem~\ref{thm:woh-deg-const-pw}). Again, this result implies \WOH{ness} of \ourprob{} when parameterized by the maximum degree of both graphs, when treewidth of both graphs is $1$. Therefore, \FPT algorithms with respect to pathwidth or treewidth as parameters are only expected to exist when the maximum degree of at least one of the graphs is constant. In Theorem~\ref{thm:tw-xp-algo}, we provide an \FPT algorithm to solve \ourprob{} when parameterized by the maximum treewidth of both input graphs and the maximum degree of one of the graphs is constant. As a consequence, the result also holds for the parameter of pathwidth of both input graphs given that one of the graphs has constant maximum degree. These structural results and their implications for \ourprob{} are listed in Table~\ref{tbl:results}. 

So far in the discussion of \ourprob{}, we look into the existence of a common star forest subgraph of size \emph{at least} $h$. However, a common star forest subgraph of size at least $h$ does not imply the existence of a common star forest subgraph of size \emph{exactly} $h$. In the full version, we show that deciding the existence of a common star forest subgraph of size at least $h$ is polynomial-time equivalent to deciding the existence of one of size exactly $h$. Therefore, all algorithmic and hardness results also hold for the problem of the existence of a common star forest subgraph of size exactly $h$. A detailed discussion can be found in the full version.

As a side note, we design an EPTAS in the full version when both input graphs are planar graphs with bounded maximum degree. This complements the \NP-hardness of \ourprob{} on such graphs as shown in Theorem~\ref{thm:3-planar-nph}. 

Due to paucity of space, some results and proofs (marked with a $\star$) have been moved to the full version.

\begin{table}
\centering
\caption{Tractability of \ourprob{} under various parameters}\label{tbl:results}
\setlength{\tabcolsep}{2pt}
\renewcommand{\arraystretch}{1.25}
\begin{tabular}{|p{0.23\textwidth}|p{0.19\textwidth}|p{0.24\textwidth}|p{0.24\textwidth}|}
 \hline
 \textbf{\emph{Bounds on both graphs}}& \textbf{Arbitrary max degree}& \textbf{Max degree $\Delta$ is a parameter} & \textbf{Max degree $\Delta$ is a constant}
 \\\hline\hline \textbf{vertex cover $\le k$} & \cellcolor{green!15} $f(k)n^{\OO(1)}$ \newline Theorem~\ref{thm:FPT-vc} & \cellcolor{green!15} $f(k)n^{\OO(1)}$ \newline Theorem~\ref{thm:FPT-vc}& \cellcolor{green!15} $f(k)n^{\OO(1)}$ \newline Theorem~\ref{thm:FPT-vc}
 \\\hline \textbf{treedepth $\le \td$} & \cellcolor{red!15} \NPH \newline for~$\td = 3$ \newline Theorem~\ref{thm:td-paraNP} & \cellcolor{green!15} $f(\td, \Delta)n^{\OO(1)}$ \newline Corollary~\ref{cor:FPT-td-maxd} & \cellcolor{green!15} $f(\td, \Delta)n^{\OO(1)}$ \newline Corollary~\ref{cor:FPT-td-maxd}
 \\\hline \textbf{pathwidth $\le \pw$} & \cellcolor{red!15} \NPH \newline for~$\pw = 2$ \newline Theorem~\ref{thm:td-paraNP} & \cellcolor{orange!20} \WOH wrt $\Delta$ \newline for~$\pw = 4$ \newline Theorem~\ref{thm:woh-deg-const-pw} & \cellcolor{green!50!red!15} $f(\pw, \Delta)n^{\OO(\Delta)}$ \newline Theorem~\ref{thm:tw-xp-algo}
 \\\hline \textbf{treewidth $\le \tw$} & \cellcolor{red!15} \NPH \newline for~$\tw = 1$ \newline Theorem~\ref{thm:td-paraNP} & \cellcolor{orange!20} \WOH wrt $\Delta$ \newline for~$\tw = 1$ \newline Theorem~\ref{thm:woh-deg-const-pw} & \cellcolor{green!50!red!15} $f(\tw, \Delta)n^{\OO(\Delta)}$ \newline Theorem~\ref{thm:tw-xp-algo}
 \\\hline \textbf{arbitrary graphs} & \cellcolor{red!15} \NPH \newline Theorem~\ref{thm:td-paraNP} & \cellcolor{red!15} \NPH for~$\Delta=3$ \newline Theorem~\ref{thm:3-planar-nph} & \cellcolor{red!15} \NPH for~$\Delta=3$ \newline Theorem~\ref{thm:3-planar-nph}\\\hline
\end{tabular}
\end{table}
 \section{Preliminaries}

\textbf{Notations.}
We denote the set $\{1,2,\ldots\}$ of natural numbers with \NB. For any integer \el, we denote the sets $\{1,\ldots,\el\}$ by $[\el]$. All graphs in this paper are undirected, unweighted, and simple. For a graph $G$, we denote the set of vertices as $V(G)$ and the set of edges as $E(G)$. The set $N_G(v) = \{u \mid \{u,v\} \in E(G)\}$ denotes the open neighbourhood of vertex $v$ in the graph $G$. The degree of a vertex $v$ in $G$ is $\deg(v) = |N_G(v)|$. A graph $G$ is called \emph{$k$-regular} if $\deg_G(v) = k$ for every vertex $v \in V(G)$. For a set of vertices $S \subseteq V$, the subgraph of $G$ induced by $S$ is denoted by $G[S]$. For a graph $G$ and a vertex $v$, let $G-v$ denote the graph $G[V(G) \setminus\{v\}]$. For a graph $G$, let $\cc(G)$ denote the set of connected components of $G$. A graph is \emph{planar} if it can be drawn in the plane without any edge crossings, and it is \emph{outerplanar} if it admits a planar drawing in which all vertices lie on the boundary of the outer face. A graph is $k$-\emph{outerplanar} if it has a planar embedding such that for every vertex, there is an alternating sequence of at most $k$ faces and $k$ vertices of the embedding, starting with the unbounded face and ending with the vertex, in which each consecutive face and vertex are incident to each other. A graph $H$ is a \emph{subgraph} of a graph $G$ (denoted by $H \subseteq G$) if there exists an injective mapping $\phi : V(H) \to V(G)$ such that for all $\{u,v\} \in E(H)$, $\{\phi(u), \phi(v)\} \in E(G)$. We call $\phi$ an \emph{embedding} of $H$ into $G$.

\begin{observation}\label{obs:G2-sub-G1}
   An instance $(G_1, G_2, h)$ of \ourprob{}, where $G_2$ is a star forest on $h$ vertices, is a \yes-instance if and only if $G_2$ is a subgraph of $G_1$.
\end{observation}
\begin{proof}
    If $G_2$ is a subgraph of $G_1$, then both $G_1$ and $G_2$ have subgraphs isomorphic to $G_2$. Moreover, since the only star forest subgraph of $G_2$ with $|V(G_2)| = h$ vertices is the entire graph $G_2$, and since $G_1$ and $G_2$ have a common star forest on $h$ vertices, $G_2$ must be a subgraph of $G_1$. 
\end{proof}

\textbf{Parameterized Complexity.} In decision problems where the input has size $n$ and is associated with a parameter $k$, the objective in parameterized complexity is to design algorithms with running time $f(k) \cdot n^{\OO(1)}$, where $f$ is a computable function that depends only on $k$ \cite{downey2012parameterized}. Problems that admit such algorithms are said to be \emph{fixed-parameter tractable} (\FPT). An algorithm with running time $f(k) \cdot n^{\OO(1)}$ is called an \FPT algorithm, and the corresponding running time is referred to as an \FPT running time \cite{cygan2015parameterized,flum2006parameterized}. 

For any graph $G$ let $\tw(G)$, $\pw(G)$, $\td(G)$ and $\tau(G)$ be its treewidth, pathwidth, treedepth and vertex cover number respectively (see the full version for definitions). Then, the structural parameters satisfy $\mathrm{tw}(G) \le \mathrm{pw}(G) \le \mathrm{td}(G) - 1 \le \tau(G)$ (see Gima et al.~\cite{gima2022exploring} for a discussion).

 \section{Classical Hardness \& Parameterization by Target Size}\label{sec:fpt-nph}
We start by showing that \ourprob{} is \NPH in the classical setting. In fact, we show a stronger result, proving \NP-hardness when both input graphs are planar and have a maximum degree of at most $3$.

\begin{theorem}[$\star$]\label{thm:3-planar-nph} 
    \ourprob{} is \NPH even when both the input graphs are planar and have a maximum degree of at most $3$.
\end{theorem}
\begin{proof}[Sketch]
    Define $P_3$ to be the path graph on $3$ vertices. Now consider the problem of {\sc $P_3$-Factor}. It asks if a graph $G$ contains a collection of vertex-disjoint subgraphs which spans all nodes in $G$, such that each subgraph is isomorphic to $P_3$. Xi and Lin~\cite{xi2021maximum} show that {\sc $P_3$-Factor} is \NPH even when $G$ is a claw-free $3$-regular planar graph. We provide a sketch of a reduction from this (the full proof can be found in the full version). Let $G$ be an instance of {\sc $P_3$-Factor} which is a $3$-regular planar graph on $n = |V(G)| = 3r$ vertices, for some integer $r$. We construct an instance $(G_1, G_2, h)$ of \ourprob{} as follows. We set $G_1 = G$, we set $G_2$ to be a disjoint collection of $r = n/3$ many copies of $P_3$, and we set the target to be $h = n$. Notice that $G_1$ and $G_2$ are both planar of maximum degree $\le 3$. The remainder of the full proof shows that $G$ is a \yes-instance of {\sc $P_3$-Factor} if and only if $(G_1, G_2, h)$ is a \yes-instance of \ourprob{}.
\end{proof}

\begin{remark}\label{rmk:EPTAS}
    We complement the \NP-hardness result when input graphs are both bounded degree planar graphs, by designing an EPTAS. For every $0 < \varepsilon < 1$, there is a $(1-\varepsilon)$-approximation algorithm, for \ourprob{}, on planar graphs with maximum degree at most $\Delta$, with running time $(1/\varepsilon + \Delta)^{\OO(1/\varepsilon)} \cdot n^{\Delta + 1} \cdot\log n$. The detailed proof is given in the full version.
\end{remark}

This motivates us to investigate the parameterized complexity of \ourprob{}. We first examine the objective as a parameter and show that {\ourprob} parameterized by the common subgraph size $h$ admits an \FPT algorithm. In contrast to this, note that the problems of \mcis, \mces, and \mcics generalise {\sc Maximum Clique} and hence are \WOH when parameterized by the objective. 

\begin{theorem}[$\star$]\label{thm:FPT-objective}
    There exists an \FPT algorithm for {\ourprob} parameterized by $h$, running in time $2^{\OO(h)} \cdot n$.
\end{theorem}

\begin{proof}[Sketch]
Let $(G_1, G_2, h)$ be an instance of \ourprob. Let $H$ be a maximum common star forest subgraph of $G_1$ and $G_2$. We reduce the problem to checking whether one of the following conditions hold: (i) both $G_1$ and $G_2$ contain a matching of size $\lceil h / 2\rceil$ (which can be checked in $2^{\OO(h)}\cdot n$ time), or (ii) for some star forest $H'$ on exactly $h$ vertices, $H'$ is a subgraph of both $G_1$ and $G_2$. Since a star forest is characterized up to isomorphism by its component sizes, such forests correspond bijectively to integer partitions of $h$ avoiding components of size $1$, giving $2^{\OO(\sqrt{h})}$ candidates~\cite{andrews1998theory}. Checking each potential $H'$ takes $2^{\OO(h)} \cdot n$ time~\cite{alon1995color}, yielding a total running time of $2^{\OO(\sqrt{h})} \cdot 2^{\OO(h)} \cdot n = 2^{\OO(h)} \cdot n$. Full details of the proof appear in the full version.\end{proof}

 \section{Structural Parameters of the Input Graphs}
In this section, we study {\ourprob} parameterized by structural parameters of $G_1$ and $G_2$. In particular we look at the structural graph parameters of vertex cover number, treedepth, pathwidth, and treewidth, in combination with the maximum degree. Recall that if a graph $G$ has treewidth $\tw(G)$, pathwidth $\pw(G)$, treedepth $\td(G)$ and vertex cover number $\tau(G)$, then it is well known~\cite{gima2022exploring} that $\tw(G) \le \pw(G) \le \td(G) - 1\le \tau(G)$.

\subsection{Vertex Cover Number as a Parameter}\label{subsec:VCmaxdeg}
We start by showing that \ourprob{} is \WTH when the parameter is the size of the minimum vertex cover in one of the graphs.
\begin{theorem}[$\star$]\label{thm:wth-vc}
{\ourprob} is \WTH when parameterized by the minimum of the vertex covers of the two graphs.
\end{theorem}

As {\ourprob} is \WTH when parameterized by the vertex cover of one graph, we discuss the case when we parameterize the problem by the maximum vertex cover number of \emph{both} input graphs. We find that {\ourprob} admits an \FPT algorithm when parameterized by the maximum vertex cover number of both input graphs.
\begin{theorem}[$\star$]\label{thm:FPT-vc}
There exists an \FPT algorithm to determine the size of the largest common star forest subgraph of two input graphs $G_1$ and $G_2$, with parameter $k$ such that both $G_1$ and $G_2$ have a vertex cover of size $\leq k$ running in time $2^{2^{\OO(k)}} \cdot n^{\OO(1)}$.
\end{theorem}

\subsection{Other Structural Parameters}\label{subsec:other}

First, since treedepth, pathwidth, and treewidth of a graph are upper-bounded by the size of a minimum vertex cover, \ourprob{} remains \WTH when parameterized by such a parameter with respect to one of the input graphs (Theorem~\ref{thm:wth-vc}). On the other hand, since \ourprob{} admits an \FPT algorithm when parameterized by the maximum vertex cover number of both input graphs, we study the parameterized complexity of the problem when parameterized by the maximum treedepth, pathwidth, or treewidth of both input graphs. We show that \ourprob{} is \PNPH for such parameters. First, we state a result due to Bodlaender et al.~\cite[Theorem 3.2]{bodlaender2020subgraph} as the following proposition.

\begin{proposition}\label{prop:siso-hard-td}
    \siso is \NPC when the host graph is a forest without paths of $6$ vertices, and the pattern is a collection of stars.
\end{proposition}

This implies the following.

\begin{theorem}\label{thm:td-paraNP} 
    \ourprob{} is \NPH even when both graphs $G_1$ and $G_2$ are collections of disjoint trees of treedepth at most $3$. 
\end{theorem}
\begin{proof}
    We reduce from the problem in Proposition~\ref{prop:siso-hard-td}. Consider the following instance of \siso: a forest $G_1$ without a path of $6$ vertices, and a collection of stars $G_2$. $G_2$ appears as a subgraph in $G_1$ if and only if they have a common star forest on at least $|V(G_2)|$ vertices; i.e. the \ourprob{} instance $G_1$, $G_2$ with target $|V(G_2)|$ is a \yes-instance. Moreover, $G_1$ and $G_2$ have treedepths at most $3$. This completes the proof.\end{proof}

Since trees of treedepth $3$ have pathwidth at most $2$ and treewidth $1$, \ourprob{} stays para-\NPH even when the treewidth or pathwidth of both input graphs is the parameter. 

Recall that Theorem~\ref{thm:3-planar-nph} shows \PNPH{ness} when the maximum degree of both input graphs is the parameter. Moreover, in this section, Theorem~\ref{thm:td-paraNP} shows \PNPH{ness} with structural parameters of treedepth, pathwidth, and treewidth. This naturally motivates us to investigate the parameterized complexity of \ourprob{} when we consider \emph{both} the maximum degree and structural parameters (like treedepth, pathwidth, and treewidth) as our combined parameter. 

In fact, we find that when we consider the treedepth and the maximum degree of both input graphs together as parameters, \ourprob{} admits an \FPT algorithm. To prove this, we first prove Theorem~\ref{thm:FPT-conn-comp}, showing that \ourprob{} admits an \FPT algorithm when the size of the largest connected component among both input graphs is the parameter. Then we use this to show Corollary~\ref{cor:FPT-td-maxd}, stating that \ourprob{} admits an \FPT algorithm when parameterized by treedepth and maximum degree of both input graphs. The strategy is similar to the proof of Theorem~\ref{thm:FPT-vc}.

\begin{theorem}[$\star$]\label{thm:FPT-conn-comp}
    There exists an \FPT algorithm for {\ourprob}, parameterized by $k$, running in time $2^{2^{\OO(k^2)}} \cdot n^{\OO(1)}$, where $k$ is the maximum size of a connected component in either input graph.
\end{theorem}
Now, we are ready to prove our result with respect to the treedepth and maximum degree of both input graphs as parameters.

\begin{corollary}\label{cor:FPT-td-maxd}
    There exists an \FPT algorithm for {\ourprob}, parameterized by $d$ and $\Delta$, running in time $2^{2^{\OO(\Delta^{2d})}} \cdot n^{\OO(1)}$, where $d$ and $\Delta$ are upper bounds on the treedepth and maximum degree of both input graphs, respectively.
\end{corollary}
\begin{proof}
    By a known result, we know that the size of any connected component of a graph of treedepth at most $d$ and maximum degree $\Delta$ is at most $\Delta^d - 1$. Therefore, the \FPT algorithm from Theorem~\ref{thm:FPT-conn-comp} directly gives us an \FPT algorithm for parameters $\Delta$ and $d$ and runs in time $2 ^ {2 ^{\OO(\Delta^{2d})}} \cdot n^{\OO(1)}$.
\end{proof}

Corollary~\ref{cor:FPT-td-maxd} motivates us to try and design \FPT algorithms for \ourprob{} when parameterized by  maximum degree and pathwidth of both input graphs. However, we prove that {\ourprob} is \WOH when parameterized by the maximum degree of both graphs, even when the pathwidth of both graphs is at most $4$.
\begin{theorem}[$\star$]\label{thm:woh-deg-const-pw} 
    {\ourprob} is \WOH with respect to $\Delta$, when both graphs are collections of disjoint trees of pathwidth at most $4$ and the maximum degree of both graphs is at most $\Delta$.
\end{theorem}

Theorem~\ref{thm:woh-deg-const-pw} proves the following for the {\sc Subgraph Isomorphism} problem.

\begin{corollary}\label{cor:subgraph-iso-hard-pw4}
    {\sc Subgraph Isomorphism} is \WOH parameterized by the maximum degree $\Delta$ of the host graph where the host graph is a collection of disjoint trees of pathwidth at most $4$ and the pattern graph is a collection of stars.
\end{corollary}

Theorem~\ref{thm:woh-deg-const-pw} also proves the \WOH{ness} of \ourprob{} parameterized by the maximum degree of both graphs, even when the graphs have treewidth $1$. In contrast, if we consider the treewidth of both graphs to be parameters, but the maximum degree to be bounded by a constant, then we design an \FPT algorithm for \ourprob{}. Our approach is based on the observation that any common subgraph without isolated vertices can be decomposed into a \emph{star forest}, where each star has a centre vertex and possibly several leaves. By first designing a dynamic programming algorithm on a nice tree decomposition for finding a maximum-cardinality star forest in a single graph in Lemma~\ref{lem:treewidth-algo-part-1}, we  extend it in Theorem~\ref{thm:tw-xp-algo} to compute a maximum common subgraph of two graphs.

At a high level, each DP-table entry at a bag encodes a partial star forest formed in the subgraph induced by the vertices processed in the corresponding subtree of the decomposition. For each vertex in the bag, we record its role: unused, a leaf (together with its centre), or a centre (with its current number of attached leaves), thereby describing how it participates in the partial solution. In addition, we maintain a star-count vector that stores how many stars of each size have already been fully formed inside the subtree. 

\begin{lemma}[$\star$]\label{lem:treewidth-algo-part-1}
    There exists an \FPT algorithm parameterized by $\tw$, which, given a graph $G$ of $n$ vertices of treewidth at most $\tw$ and constant maximum degree $\Delta$, enumerates all star forests $H$ (unique up to isomorphism) appearing as a subgraph in $G$ with the maximum degree of at most $\Delta$, running in time $\OO\left((\Delta + \tw)^{2(\tw+1)} \cdot \tw \cdot \Delta \cdot n^{\Delta+1} \cdot \log n \right)$.
\end{lemma}

We now extend the dynamic programming algorithm of Lemma~\ref{lem:treewidth-algo-part-1} to compute a maximum common subgraph of two graphs that contains no isolated vertices.

\begin{theorem}\label{thm:tw-xp-algo}
There exists an \FPT algorithm for {\ourprob}, parameterized by $\tw$ when the maximum degree of one of the input graphs is constant, running in time
\[
    \OO\left((\Delta + \tw)^{2(\tw+1)} \cdot \tw \cdot \Delta \cdot n^{\Delta+1} \cdot \log n\right),
\]
where $n = \max(|V(G_1)|,|V(G_2)|)$, provided $\tw(G_1) \leq \tw$, $\tw(G_2) \leq \tw$, and $\Delta(G_1) \leq \Delta$ or $\Delta(G_2) \leq \Delta$ for some constant $\Delta$.
\end{theorem}

\begin{proof}
Let $H$ be a maximum common star forest subgraph of $G_1$ and $G_2$, then the maximum degree of $H$ is $\Delta$. Using Lemma~\ref{lem:treewidth-algo-part-1}, we enumerate all star forest subgraphs of $G_1$ and $G_2$ of maximum degree $\Delta$, and output the size of the largest star forest appearing as subgraphs of both.
\end{proof} \section{Conclusion}
In this paper, we obtain a complete dichotomy of the parameterized complexity of \ourprob{} when the parameters are a combination of the common subgraph size, or vertex cover, maximum degree, treedepth, pathwidth or treewidth of one or both input graphs. In future, it would be interesting to study the scope for designing efficient approximation schemes for the problems with parameterized hardness results, similar to the EPTAS we design for the problem when both input graphs are planar graphs with constant maximum degrees. Exploring tractability with respect to other structural parameters such as neighbourhood diversity, vertex integrity, cliquewidth, etc. also remains an interesting direction. From the perspective of computational social choice, we also wish to study the problem where the input graphs come from special graph classes like unit disk graphs (as they capture geometric nearness of houses on a Euclidean plane), etc.
 
\begin{credits}
\subsubsection{\ackname}
Anubhav Dhar was supported by the Deutsche Forschungsgemeinschaft (DFG, German Research Foundation) under grant number 559177164.

\subsubsection{\discintname}
The authors declare that they have no relevant or material financial interests that relate to the research described in this paper.
\end{credits}

\bibliographystyle{splncs04}
\bibliography{ref}

\appendix
\renewcommand{\theHsection}{\Alph{section}}
\section*{Appendix}

\section{Extended Preliminaries}
\subsection{Structural parameters}\label{subsec:structural-parameters}

Following are the definitions of tree decompositions, treewidth, nice tree decompositions, path decompositions, pathwidth, treedepth, and vertex cover. We start by defining tree decompositions, treewidth, and nice tree decompositions.

\begin{definition}[Tree Decomposition and Treewidth] Let $G$ be a graph.  A {\em tree-decomposition} of a graph $G$ is a pair  $(T, X=\{X_{t}\}_{t\in V(T)})$,  where  $T$ is a tree where every node $t\in V(T)$ is assigned a subset $X_t\subseteq V(G)$, called a \emph{bag},  such that the following conditions hold. 
\begin{enumerate}
\item $\bigcup\limits_{t\in V(T)}{X_t}=V(G)$,
\item for every edge $\{x,y\}\in E(G)$ there is a $t\in V(T)$ such that  $x,y\in X_{t}$, and 
\item for any $v\in V(G)$ the subgraph of $T$ induced by the set  $\{t \in V(T)\mid v\in X_{t}\}$ is connected.
\end{enumerate}
The {\em width} of a tree decomposition is $\max\limits_{t\in V(T)} |X_t| -1$. The {\em treewidth} of $G$ is the minimum width over all tree decompositions of $G$ and is denoted by $\tw(G)$.
\end{definition}

\begin{definition}[Nice tree decomposition]A tree decomposition $(T,X)$ is called a {\em nice tree decomposition} if $T$ is a tree rooted at some node $r$ where $ X_{r}=\emptyset$, each node of $T$ has at most two children, and each node is one of the following kinds:
\begin{itemize}
\item {\bf Introduce node:} a node $t$ that has only one child $t'$ where $X_{t}\supset X_{t'}$ and  $|X_{t}|=|X_{t'}|+1$. 
\item {\bf  Forget node:} a node $t$ that has only one child $t'$  where $X_{t}\subset X_{t'}$ and  $|X_{t}|=|X_{t'}|-1$.
\item {\bf Join node:} a node  $t$ that has two children $t_{1}$ and $t_{2}$ such that $X_{t}=X_{t_{1}}=X_{t_{2}}$.
\item {\bf Leaf node:} a node $t$ that is a leaf of $T$, and $X_{t}=\emptyset$. 
\end{itemize} 
\end{definition}

A tree decomposition $T$ of width $\tw$ can be transformed into a nice tree decomposition of the same width with $\OO(\tw \cdot |V(G)|)$ nodes in time $\OO(\tw^2\cdot \max(|V(T)|,|V(G)|))$~\cite{cygan2015parameterized}.

The following is known regarding the treewidth of $k$-outerplanar graphs~\cite{williamson2011design}.

\begin{proposition}\label{prop:k-outerplanar-tw}
    A $k$-outerplanar graph has treewidth at most $3k+1$.
\end{proposition}

Pathwidth is a similar parameter to treewidth, except when the underlying graph is a path. Treedepth is yet another structural parameter. These are defined below.

\begin{definition}[Path Decomposition and Pathwidth]
A path decomposition of a graph $G$ is a sequence
$P = (X_1, X_2, \ldots, X_r)$ of bags, where $X_i \subseteq V(G)$ for each
$i \in \{1,2,\ldots,r\}$, such that the following conditions hold:
\begin{enumerate}
    \item $\bigcup\limits_{i\in[r]} X_i = V(G)$,
    
    \item $\forall \{u,v\} \in E(G)$, $\exists \ell \in \{1,2,\ldots,r\}$ such that the bag $X_\ell$ contains both $u$ and $v$.
    
    \item $\forall u \in V(G)$, if $u \in X_i \cap X_k$ for some $i \le k$, then $u \in X_j$ for all $j$, with $i \le j \le k$.
\end{enumerate}
The width of a path decomposition $(X_1, X_2, \ldots, X_r)$ is
$\max\limits_{i \in [r]} |X_i| - 1$.
The pathwidth of a graph $G$, denoted by $\mathrm{pw}(G)$, is the minimum possible width of a path decomposition of $G$.
\end{definition}

\begin{definition}[Treedepth]The treedepth of $G$, denoted by $\td(G)$, is defined as follows:

\[
  \td(G) = \begin{cases}
    0 & \text{if $G$ has no vertices,} \\
    \max \limits_{H \in \cc(G)} \{\td(H)\} & \text{if $G$ is disconnected,} \\
    \min \limits_{v \in V(G)} \{\td(G - v)\} + 1& \text{if $G$ is connected.}
  \end{cases}
\]
\end{definition}

We prove the following folklore results related to treedepth for completeness.

\begin{lemma}\label{lem:td-maxd-cc}
The size of any connected component of a graph of treedepth at most $d$ and maximum degree $\Delta$ is at most $\Delta^d - 1$.
\end{lemma}
\begin{proof}
    We proceed by induction on $d$. For $d = 0$, the graph must have no vertices, in particular the number of vertices is at most $\Delta^d - 1 = 0$. Now, for some $d \ge 1$, assume that all graphs of maximum degree at most $\Delta$ and treedepth at most $d-1$ have connected components of size at most $\Delta^{d-1} - 1$. Let $G$ be a graph of maximum degree at most $\Delta$ and tree depth at most $d$. Every connected component $C$ of $G$ must have maximum degree at most $\Delta$ and treedepth at most $d$. We show that $|C| \le \Delta^d - 1$.
    
    By definition, $d \ge \td(C) = \min \limits_{v \in V(C)}\{\td(C-v)\} + 1$. Let $u \in V(C)$ be the vertex such that $\td(C) = \td(C - u) + 1$. Therefore, every connected component of $C - u$ must have treedepth at most $d - 1$ and there can be no larger than $\Delta^{d-1} - 1$. Moreover, these can be at most $\Delta$ many connected components in $C - u$ as the degree of $u$ in $C$ is at most $\Delta$. Therefore $|C - u| \le \Delta \cdot (\Delta ^ {d - 1} - 1) = \Delta^d - \Delta$, and hence $|C| = |C- u| + 1\le \Delta^d - 1$.
\end{proof}
    
\begin{lemma}\label{lem:diam-d-td-d}
    The treedepth of a forest having no path on $2\ell$ nodes, is at most $\ell$.
\end{lemma}
\begin{proof}
    Note that it suffices to prove that no tree in the forest has treedepth more than $\ell$. First, we prove the following claim. 
    \begin{claim}\label{clm:centre-td}
        For a tree $T$, if there exists a vertex $v$ in it, such that the farthest vertex $u$ from $v$ be such that the path from $v$ to $u$ has $\le p$ nodes, then the treewidth $\td(T) \le p$.
    \end{claim}
    \begin{proof}
        We proceed via induction. Clearly, for $p = 1$, $T$ has exactly one node, hence $\td(T) \le 1$.
        
        Now consider $p \ge 2$, and assume that the claim holds true for $p - 1$. Let $v_1, \ldots, v_k$ be the neighbours of $v$. Let $T_1, \ldots, T_k$ be the connected components of $T - v$ such that $v_i \in V(T_i)$, for all $i \in k$. Since the farthest node $u$ from $v$ is such that the path from $u$ to $v$ has $\le p$ nodes, the farthest node $u_i \in V(T_i)$ from $v_i$ must be such that the path from $u_i$ to $v_i$ has at most $(p-1)$ nodes. Therefore $\td(T_i) \le p - 1$ by the induction hypothesis. This implies,
        $\td(T) \le \td(T - v) + 1 = \max \limits_{H \in \cc(T - v)} \{\td(H)\} + 1 = \max_{i \in [k]} \{\td(T_i)\} \ + 1 \le p - 1 + 1 = p.$
    \end{proof}
    
    Consider a tree $T$ in the forest. Let $v$ be the centre~\cite{west2001introduction} of the tree $T$. Since there are no paths of $2\ell$ nodes, the farthest vertex $u \in T$ from $v$ is such that the path from $u$ to $v$ contains at most $\ell$ nodes. From Claim~\ref{clm:centre-td}, we have $\td(T) \le \ell$. \end{proof}

Finally, we recall the definition of a vertex cover.

\begin{definition}[Vertex Cover]
Let $G$ be a graph. A set $S \subseteq V(G)$ is a \emph{vertex cover} of $G$ if the graph $G - S \coloneq G[V(G) \setminus S]$ is edgeless. The vertex cover number, denoted as \(\tau (G)\), is the size of the minimum vertex cover.
\end{definition}

For a graph $G$, the structural parameters satisfy $\mathrm{tw}(G) \le \mathrm{pw}(G) \le \mathrm{td}(G) - 1 \le \tau(G)$ (see Gima et al.~\cite{kobayashi2017treedepth} for a discussion). For more details regarding parameterized complexity, refer to Cygan et al~\cite{cygan2015parameterized}.

\subsection{Approximation algorithms} 
An algorithm \AA for a problem \PP is said to be efficient if it solves \textit{all instances} of the problem \PP \textit{optimally} in \textit{polynomial time}. However, unless \Pb = \NP, we cannot hope to have efficient algorithms for the class of problems belonging to \NPH. A natural way out is to choose any two from the three requirements. In the realm of approximation algorithms, we relax the requirement of \textit{optimality} \cite{williamson2011design}. 

\begin{definition}
    An $\alpha$-approximation algorithm for an optimization problem is a polynomial time algorithm that for all instances of the problem produces a solution whose value is within a factor of $\alpha$ of the value of an optimal solution.
\end{definition}

We call $\alpha$ as the approximation ratio. Typically, for minimization problems, $\alpha > 1$ and for maximization problem, $\alpha  < 1$ \cite{williamson2011design}.

\begin{definition}[PTAS]
A polynomial time approximation scheme (PTAS) is a family of algorithms $\AA_\varepsilon$, where there is an algorithm for each $\varepsilon > 0$, such that $\AA_\varepsilon$ is a $(1 + \varepsilon)$ approximation algorithm (for minimization problems) or a $(1 - \varepsilon)$ approximation algorithm (for maximization problems), and the running time of $\AA_\varepsilon$ is polynomial in the input size for every constant $\varepsilon$.
\end{definition}

\begin{definition}[EPTAS] An efficient polynomial time approximation scheme (EPTAS) is a family of algorithms $\AA_\varepsilon$, where there is an algorithm for each $\varepsilon > 0$, such that $\AA_\varepsilon$ is a $(1 + \varepsilon)$ approximation algorithm (for minimization problems) or a $(1 - \varepsilon)$ approximation algorithm (for maximization problems), and the running time of $\AA_\varepsilon$ is $f(1/\varepsilon)\cdot n^{\OO(1)}$, where $f$ is a computable function depending only on $\varepsilon$, and $n$ is the input size.
\end{definition}

\subsection{Other algorithmic and combinatorial tools}
We start by mentioning the following algorithm for enumerating partitions $P(h)$ of a given integer $h$.

\begin{lemma}\label{lem:part-algo}
    There exists an algorithm which enumerates all partitions of an integer $h$ in time $2^{\OO(\sqrt h)}$.
\end{lemma}
\begin{proof}
    Consider the recursive procedure described in Algorithm~\ref{alg:part}.
    \begin{algorithm}[!htbp]
            \caption{ \hfill \textbf{Input:} $h$ \hfill \textbf{Output:} $P(h)$}\label{alg:part}
            \begin{algorithmic}[1]
                \Procedure{Partition}{multiset $S$, integer $h$}:
                    \State{$\mathtt{sum} \gets \sum_{a \in S}a$}
                    \State{Output $S \cup \{h - \mathtt{sum}\}$}
                    \For{$i$ from $\max(S \cup \{1\})$ to $(h - \mathtt{sum} )/ 2$}
                        \State{Run \Call{Partition}{$S \cup \{i\}$, $h$}}
                    \EndFor 
                \EndProcedure
                \State{Run \Call{Partition}{$\emptyset$, $h$}}
            \end{algorithmic}
        \end{algorithm}
    Algorithm~\ref{alg:part} outputs all partitions uniquely, and since no branch of the algorithm faces a dead end, there are exactly $P(h)$ leaves of the recursion tree. Moreover, the height of the recursion tree is at most $h$ and each execution takes $\OO(h)$ time. Therefore the entire algorithm terminates in time $\OO(h^2 \cdot 2^{\OO(\sqrt h)}) = 2^{\OO(\sqrt h)}$.
\end{proof}

Next, we state the following result for a minimum dominating set of a graph.
    
\begin{lemma}\label{lem:dom-match}
    For a graph $G$ without isolated vertices, if $D$ is a minimum dominating set, then there exists a matching $M$ exhausting $D$, consisting of edges which are incident on one vertex in $D$ and another vertex outside $D$. 
\end{lemma}
\begin{proof}
    Let $D$ be a minimum dominating set of $G$. Consider the bipartite graph $G'$ induced by edges with one endpoint in $D$ and one endpoint outside $D$. 

    First we show that $G'$ has no isolated vertex. Vertices outside $D$ have at least one neighbour in $D$, and hence they cannot be isolated. Now, let $v \in D$ be an isolated vertex in $G'$. $v$ is not isolated in $G$ (as $G$ had no isolated vertices to start with). Thus, $N_G(v) \subset D$. Therefore, $D \setminus \{v\}$ is also a dominating set of $G$, which contradicts the minimality of $D$.

    Now, $D$ must be a minimum dominating set of $G'$, because any dominating set of $G'$ is also a dominating set of $G$. Moreover, $D$ is a vertex cover of $G'$. Since any vertex cover is a dominating set, $D$ must be a minimum vertex cover of $G'$. By Kőnig's theorem~\cite{konig1916graphen}, the maximum matching $M$ of $G'$ is of cardinality $|D|$, hence it exhausts $D$.
\end{proof}

The next result is a well-known technique to find a sumset, i.e. the set of all pairwise sums of elements from two sets, using a fast polynomial multiplication algorithm. Such a result has been discussed and used throughout literature~\cite{bringmann2017near,jin2019simple,fischer2025sumsets}. We state a special use case of the celebrated Fast Fourier Transform (FFT) algorithm for polynomial multiplication~\cite{cormen2022introduction}.

\begin{proposition}[FFT]\label{prop:fft}
    Given two polynomials $f(x), g(x)$ with coefficients in $\{0, 1\}$ and of degree at most $t$, there exists an algorithm that computes $f(x) \cdot g(x)$ in time $\OO(t \log t)$. 
\end{proposition}

We use this to get the following result.

\begin{lemma}\label{lem:sumset}
    Given two sets $A, B \subseteq \{0, 1, \ldots, n\}^d$, there exists a $\OO(dn^d \log n)$-time algorithm to compute   
    \[A \oplus B = \{(a_1 + b_1, a_2 + b_2, \ldots, a_d + b_d) \mid (a_1, a_2 \ldots, a_d) \in A, (b_1, b_2, \ldots, b_d) \in B\}.\]
\end{lemma}
\begin{proof}
    For $\mathbf{s} = (s_1, \ldots, s_d) \in \{0, 1, \ldots, 2n\}^d$, let $\gamma_\mathbf{s} = \sum_{i \in [d]} s_i \cdot (2n+1)^{i - 1}$.

    \begin{claim}\label{clm:base-2n+1}
        Let $\mathbf{a}, \mathbf{b} \in \{0, 1, \ldots, n\}^d$ and $\mathbf{c} \in \{0, 1, \ldots, 2n\}^d$. Then $\gamma_\mathbf{a} + \gamma_\mathbf{b} = \gamma_\mathbf{c}$ if and only if $a_i + b_i = c_i$ for every $i \in [d]$.
    \end{claim}
    \begin{proof}
        If $a_i + b_i = c_i$ for every $i \in [d]$, then 
        \[\gamma_\mathbf{a} + \gamma_\mathbf{b} = \sum \limits_{i \in [d]} a_i \cdot (2n+1)^{i - 1} + \sum \limits_{i \in [d]} b_i \cdot (2n+1)^{i - 1} = \sum \limits_{i \in [d]} (a_i + b_i) \cdot (2n+1)^{i - 1} = \gamma_\mathbf{c}.\]
        Conversely, suppose $\gamma_\mathbf{a} + \gamma_\mathbf{b} = \gamma_\mathbf{c} = \gamma$. For all $i \in [d]$, $(a_i + b_i), c_i \le 2n$, and therefore $\gamma < (2n+1)^d$. Since the base $(2n+1)$ representation of $\gamma = \sum_{i \in [d]} \beta_i (2n+1)^{i-1}$ for $\beta_i \in \{0, 1, \ldots, 2n\}$, $i \in [d]$ is unique, we must have $a_i + b_i = \beta_i = c_i$ for every $i \in [d]$.
    \end{proof}

    Consider the polynomials $p(x) = \sum_{\mathbf{a} \in A} x^{\gamma_\mathbf{a}}$, and $q(x) = \sum_{\mathbf{b} \in B} x^{\gamma_\mathbf{b}}$. $p(x)$ and $q(x)$ are polynomials with coefficients in $\{0,1\}$ and have degree at most $(2n + 1)^d$. Due to Claim~\ref{clm:base-2n+1}, for $\mathbf{c} \in \{0,1,\ldots, 2n\}^d$, $x^{\gamma_\mathbf{c}}$ has a non zero coefficient in $p(x) \cdot q(x)$, if and only if $\mathbf{c} \in A \oplus B$. Therefore, $A \oplus B$ can be constructed by first finding $p(x) \cdot q(x)$ in time $\OO((n^d) \log(n^d)) = \OO(dn^d \log n)$ (Proposition~\ref{prop:fft}), and then taking all $i$ such that $x^i$ has a non-zero coefficient in $p(x) \cdot q(x)$.
\end{proof}

\section{Equivalence in Finding Common Subgraphs of size exactly $h$ and at least $h$}\label{subsec:eqv-eq-ge}

Consider the following problem.
\defproblem{Common Star Forest Subgraph of Target Size (CSFSTS)}{Undirected unweighted graphs $G_1$, $G_2$, non-negative integer $h$.}{Does there exist a star forest $H$ on exactly $h$ vertices, such that $H$ is a subgraph of both $G_1$ and $G_2$?}

In this section, we show that \ourprob{} and CSFSTS are equivalent under polynomial-time reductions. The following Lemmas captures the primary essence of the equivalence.

\begin{lemma}\label{lem:eq-or-matching}
    Let $H$ be a maximum common star forest subgraph of $G_1$ and $G_2$, then $|V(H)| \ge h$ if and only if at least one of the following holds true.
    \begin{enumerate}
        \item $G_1$ and $G_2$ both have a matching consisting of $\lceil h/2 \rceil$ edges.
        \item $G_1$ and $G_2$ both have a common star forest subgraph of size exactly $h$.
    \end{enumerate}
\end{lemma}
\begin{proof}
    Clearly, if $G_1$ and $G_2$ both have a matching (which is also a star forest) consisting of $\lceil h/2 \rceil$ edges or if they have a common star forest of size exactly $h$, then the maximum common star forest would be at least of size $h$, i.e. $|V(H)| \ge h$.

    Now, we prove the other direction. Let $H$ be a common star forest subgraph of $G_1$ and $G_2$ with $|V(H)| \ge h$, and let the size of the maximum matching in either $G_1$ and $G_2$ be less than $\lceil h/2 \rceil$. Therefore, $\lceil h/2 \rceil$ is strictly larger than the size of the maximum matching of $H$. Hence, $H$ can have at most $\lceil h/2 \rceil - 1$ many stars. Let the stars of $H$ be $S_1, \ldots, S_t$ for $t \le \lceil h/2 \rceil - 1$. For every $i \in [t]$, let $S_i$ have centre $c_i$ and leaves $v_i^{(0)}, v_i^{(1)}, \ldots, v_i^{(y_i)}$, for some $y_i \ge 0$. Notice that the graph induced by $Q = \{c_1, c_2, \ldots, c_t\} \cup \{v_1^{(0)}, v_2^{(0)}, \ldots, v_t^{(0)}\}$ induce a matching of size $t$ in $H$, containing $2t \le 2\lceil h/2 \rceil - 2 < h$ nodes. Now, we take all vertices from $Q$ and we take $h - 2t$ vertices from $V(H) \setminus Q$ to get a subset $R$ of $V(H)$ of size exactly $h$. $H[R] \subseteq H$ is a subgraph of both $G_1$ and $G_2$ of size $h$. Moreover, $H[R]$ is a star forest as it contains no isolated vertices; this is because the centres of stars $S_1, \ldots, S_t$, i.e. $c_1, \ldots, c_t$ are present in $R$.
\end{proof}

\begin{lemma}\label{lem:vert-deg-atleast-2}
    Let $G_1$ and $G_2$ be graphs of maximum degree at least $2$. Then for an odd integer $h \ge 3$, the following statements are equivalent.
    \begin{enumerate}
        \item $G_1$ and $G_2$ have a common star forest subgraph on exactly $h$ vertices.
        \item $G_1$ and $G_2$ have a common star forest subgraph on at least $h$ vertices.
    \end{enumerate}
\end{lemma}
\begin{proof}
    Clearly if $G_1$ and $G_2$ have a common star forest on exactly $h$ vertices, then they have a common star forest on at least $h$ vertices. 

    We now show the other direction, suppose $G_1$ and $G_2$ have a common star forest subgraph on at least $h$ vertices. By Lemma~\ref{lem:eq-or-matching}, $G_1$ and $G_2$ either have a matching on at least $\lceil h / 2 \rceil = (h+1)/2$ edges or a common star forest subgraph of exactly $h$ vertices. If the latter is true, then we are done. So, assume that both $G_1$ and $G_2$ have a matching on at least $(h+1)/2$ edges. 
    
    Define the graph $H^*$ as the disjoint collection of $(h-3)/2$ many stars of size $2$, and one star of size $3$. We show the following claim.

    \begin{claim}\label{clm:deg2-subgraph}
        For an odd integer $h \ge 3$, if a graph $G$ with maximum degree at least $2$ has a matching on at least $(h+1)/2$ edges, then $H^* \subseteq G$.
    \end{claim}
    \begin{proof}
        Let $v \in V(G)$ be a vertex of degree at least $2$, and let $a$ and $b$ be two distinct neighbours of $v$. Let $V' = V(G) \setminus \{v, a, b\}$. Since $G$ has a matching $M$ on $(h + 1)/2$ edges, $G[V']$ must have a matching on at least $(h + 1)/2 - 3 = (h - 3)/2 - 1$ edges. We analyse the following cases. 
        \begin{itemize}
            \item \textbf{Case I: $G[V']$ has a matching containing $(h - 3)/2$ edges.} Then this matching along with the edges $\{v, a\}$ and $\{v, b\}$ form a subgraph of $G$ isomorphic to $H^*$, giving us $H^* \subseteq G$. 
            \item \textbf{Case II: $G[V']$ has no matching containing $(h - 3)/2$ edges.} Therefore, $M$ must contain all $v$, $a$ and $b$ where these vertices appear in three distinct edges of $M$. Then there must exist $c \notin \{a, b\}$ such that the edge $\{v, c\}$ is in the matching $M$. Moreover, there exists $d \notin \{v, a\}$ such that the edge $\{b, d\}$ is in the matching $M$. Let $M'$ be the matching on $(h - 3)/2$ edges obtained by removing edges $\{v, c\}$ and $\{b, d\}$ from $M$. Therefore, the matching $M'$ along with the edges $\{v, b\}$ and $\{v, c\}$ form a subgraph of $G$ isomorphic to $H^*$, giving us $H^* \subseteq G$.
        \end{itemize}
        Hence in all cases, we have $H^* \subseteq G$.
    \end{proof}

    Therefore $H^* \subseteq G_1$ and $H^* \subseteq G_2$. Moreover, since $H^*$ is a star forest on $2 \cdot (h - 3)/2 + 3 \cdot 1 = h$ vertices, $G_1$ and $G_2$ have a star forest subgraph on exactly $h$ vertices. 
\end{proof}

We now present polynomial-time reductions between these problems. Let $P_2$ be the connected graph on $2$ vertices. Notice that $(P_2, P_2, 2)$ is a \yes-instance and $(P_2, P_2, 3)$ is a \no-instance of both \ourprob{} and CSFSTS.

\begin{theorem}\label{thm:prob-eqv}
    \ourprob{} and {\sc CSFSTS} are equivalent under polynomial-time reductions.
\end{theorem}
\begin{proof}
    We start by reducing from \ourprob{} to CSFSTS. Let $(G_1, G_2, h)$ be an instance of \ourprob{}. The reduction produces an instance of CSFSTS as follows. We compute the maximum matching of $G_1$ and $G_2$ in polynomial time~\cite{edmonds1965paths}. If both graphs have matchings of size at least $\lceil h/2 \rceil$ then output the trivial \yes-instance $(P_2, P_2, 2)$. Otherwise output the instance $(G_1, G_2, h)$. Lemma~\ref{lem:eq-or-matching} directly implies the correctness of this reduction.

    For the other direction, consider an instance $(G_1, G_2, h)$ of CSFSTS. The reduction produces an instance of \ourprob{} as follows. If $h$ is even, output the instance $(G_1, G_2, h)$ of \ourprob{}. If $h$ is odd and if either $G_1$ or $G_2$ has maximum degree of $1$, then output the trivial \no-instance $(P_2, P_2, 3)$. If $h = 1$, output the trivial \no-instance $(P_2, P_2, 3)$. Otherwise (i.e. if $h$ is odd and the maximum degrees of $G_1$ and $G_2$ are both at least $2$), output the instance $(G_1, G_2, h)$. We analyse the correctness of this reduction.
    \begin{itemize}
        \item If $h$ is even, then a matching on $\lceil h / 2 \rceil = h/2$ edges is also a star forest of size exactly $h$. Therefore, Lemma~\ref{lem:eq-or-matching} implies that CSFSTS and \ourprob{} are equivalent on the instance $(G_1, G_2, h)$.
        \item If $h$ is odd and either $G_1$ or $G_2$ has maximum degree $1$, then observe that $(G_1, G_2, h)$ is a \no-instance of CSFSTS. This is because any star forest of odd size must have a star with at least $3$ nodes, and therefore must have a vertex of degree at least $2$. The reduction also outputs a \no-instance of \ourprob{}.
        \item If $h = 1$, then $(G_1, G_2, h = 1)$ is a \no-instance to CSFSTS as any star forest must have size at least $2$. The reduction also outputs a \no-instance to \ourprob{}.
        \item If $h \ge 3$ is odd and $G_1$ and $G_2$ have vertices of degree at least $2$, then by Lemma~\ref{lem:vert-deg-atleast-2}, \ourprob{} and CSFSTS are equivalent on the instance $(G_1, G_2, h)$.
    \end{itemize}
    Therefore, \ourprob{} and CSFSTS are equivalent under polynomial-time reductions.
\end{proof} \section{EPTAS for Bounded Degree Planar Graphs}\label{sec:EPTAS}

In this section, we examine the problem of \ourprob{}, where the input graphs are planar graphs with maximum degree bounded by some constant. Although \ourprob{} is \NPH when both graphs are planar graphs of bounded degree (by a constant $\Delta = \OO(1)$) (Theorem~\ref{thm:3-planar-nph}), we show that the problem admits an EPTAS for such inputs. The algorithm is a modification of Baker's technique~\cite{baker1994approximation} which sacrifices a small part of optimal solution to break the graph down to $k$-outerplanar graphs, for a small $k$ and then solves the problem exactly.

\begin{theorem}\label{thm:EPTAS}
    For every $0 < \varepsilon < 1$, there is a $(1-\varepsilon)$-approximation algorithm, for \ourprob{}, on planar graphs with maximum degree at most $\Delta$, with running time $(1/\varepsilon + \Delta)^{\OO(1/\varepsilon)} \cdot n^{\Delta + 1} \cdot\log n$.
\end{theorem}

\begin{proof}
    We first state the algorithm and then prove its correctness.

    \textbf{The approximation algorithm.}
    Let $0 < \varepsilon < 1$, and the input graphs be $G_1$, $G_2$.  Consider a BFS tree $\TT$ of $G_1$. Let $\Delta$ be the largest degree in these graphs. Let $L^{(i)}_1 \subseteq V(G_1)$ be the $i$-th level of vertices in $\TT$. Similarly let $L^{(i)}_2$ be the $i$-th level of vertices in some BFS tree of $G_2$. Let $k = \lceil 2 / \varepsilon \rceil$.

    The algorithm iterates over $r_1, r_2 \in \{0, 1, \ldots, k - 1\}$. The algorithm creates the graph $G_1^{(r_1)}$ which is the graph obtained from $G_1$ by removing the vertices lying in level $L^{(3r_1 + 2)}_1, L^{(3(k + r_1) + 2)}_1, L^{(3(2k + r_1) + 2)}_1, L^{(3(3k + r_1) + 2)}_1, \ldots$. Similarly it creates the graph $G_2^{(r_2)}$ which is the graph obtained from $G_2$ by removing the vertices lying in level $L^{(3r_2 + 2)}_2, L^{(3(k + r_2) + 2)}_2, L^{(3(2k + r_2) + 2)}_2, L^{(3(3k + r_2) + 2)}_2, \ldots$. $G^{(r_1)}_1$ and $G^{(r_2)}_2$ are $3k$-outerplanar graphs and have treewidth at most $9k + 1$ (due to Proposition~\ref{prop:k-outerplanar-tw}). The algorithm now solves \ourprob{} on $G_1^{(r_1)}$ and $G_2^{(r_2)}$ in time $\OO\left((\Delta + (9k+1))^{2(9k+2)} \cdot (9k+1) \cdot \Delta \cdot n^{\Delta+1} \cdot \log n \right)$ (using Theorem~\ref{thm:tw-xp-algo}). The algorithm outputs the solution which attains the largest number of nodes, over all choices of $r_1, r_2$. The total running time of the algorithm is $k^2 \cdot \OO\left((\Delta + (9k+1))^{2(9k+2)} \cdot (9k+1) \cdot \Delta \cdot n^{\Delta+1} \cdot \log n \right) = (1/\varepsilon + \Delta)^{\OO(1/\varepsilon)} \cdot n^{\Delta + 1} \cdot \log n$.

    \textbf{Correctness: guarantee of the approximation factor.}
    Let \OPT{} be the number of nodes of the maximum common star forest subgraph $H$ of $G_1$ and $G_2$. Let \SOL{} be the number of nodes in the output by the algorithm. Clearly $\SOL < \OPT$, as the algorithm finds largest common subgraph of a subgraph of $G_1$ and a subgraph of $G_2$.

    Let $S$ be a star in a subgraph $H_1$ of $G_1$ which is isomorphic to $H$; hence $S$ is a star subgraph of $G_1$. There exists $i$ such that $V(S) \subseteq L^{(i)}_1 \cup L^{(i + 1)}_1 \cup L^{(i + 2)}_1$, as a star cannot span across four or more levels. Therefore for at most one $r_1 \in \{0, 1, \ldots, k-1\}$, $G^{(r_1)}_1$ does not contain all the vertices of $S$; in such a case we say $r_1$ \emph{stabs} $S$. Let $s_{r_1}$ denote the total number of vertices in all stars of $H_1$ stabbed by $r_1$. Let $r^*_1 \in {\argmin}_{r_1 \in \{0, \ldots, k-1\}}\{s_{r_1}\}$ minimize the number of nodes of stars in $H_1$ it stabs. Let $F_1$ be the forest of stars in $H$ corresponding to the stars stabbed by $r^*_1$, then $|F_1| \le \frac{1}{k} \cdot|V(H)|$. Moreover, $H \setminus F_1$ is a subgraph of $G_1^{(r_1^*)}$. Similarly, there exists $r^*_2 \in\{0, \ldots, k-1\}$ and a forest of stars $F_2$ in $H$, such that $|F_2| \le \frac{1}{k} \cdot|V(H)|$ and $H \setminus F_2$ is a subgraph of $G_2^{(r_2^*)}$. Therefore $H \setminus (F_1 \cup F_2)$ is a common star forest subgraph of both $G_1^{(r_1^*)}$ and $G_2^{(r_2^*)}$. Therefore, we have 
    \begin{align*}
        \SOL &\ge |V(H) \setminus (V(F_1) \cup V(F_2))| \ge |V(H)| - |V(F_1)| - |V(F_2)| \\
        &\ge \left(1 - \frac{2}{k}\right) \OPT \ge (1 - \varepsilon) \OPT.
    \end{align*}
\end{proof} \section{Missing Proofs of Section~\ref{sec:fpt-nph}}

\subsection{Proof of Theorem~\ref{thm:3-planar-nph}}
Before stating the proof of Theorem~\ref{thm:3-planar-nph}, we first state the problem of {\sc $P_3$-Factor}. Recall that $P_3$ is a path graph on $3$ vertices.

\defproblem{$P_3$-Factor}{Undirected, unweighted graph $G$}{Does $G$ contain a collection of vertex-disjoint subgraphs, which spans all nodes in $G$, in which each subgraph is isomorphic to $P_3$?}

Xi and Lin~\cite{xi2021maximum} show that {\sc $P_3$-Factor} is \NPH even when $G$ is a claw-free $3$-regular planar graph. We use this result to prove Theorem~\ref{thm:3-planar-nph}.

\begin{proof}
    We reduce {\ourprob} from {\sc $P_3$-Factor} in $3$-regular planar graphs. Let $G$ be an instance of {\sc $P_3$-Factor} which is a $3$-regular planar graph. We may assume that $n = |V(G)| = 3r$, for some integer $r$ (otherwise $G$ is a \no-instance of {\sc $P_3$-Factor}). We construct an instance $(G_1, G_2, h)$ to \ourprob{} as follows. We set $G_1 = G$, we set $G_2$ to be a disjoint collection of $n/3$ many copies of $P_3$, and we set the target to be $h = n$. Notice that $G_1$ and $G_2$ are both planar where $G_1$ has maximum degree $3$ and $G_2$ has maximum degree $2$.

    If $G$ is a \yes-instance of {\sc $P_3$-Factor}, then $G_2$ is a subgraph of $G_1 = G$ and hence $(G_1, G_2, h)$ is a \yes-instance of \ourprob{}. Conversely, if $G_1$ and $G_2$ have a common spanning star forest subgraph, then it must be isomorphic to $G_2$, as $G_2$ is a star forest itself. This means $G_1 = G$ has $G_2$ as its subgraph, making $G$ a \yes-instance for {\sc $P_3$-Factor}.
\end{proof}

\subsection{Proof of Theorem~\ref{thm:FPT-objective}}
\begin{proof}
    Let $H$ be a maximum common star forest subgraph of $G_1$ and $G_2$. We need to decide if $|V(H)| \ge h$. Recall that Lemma~\ref{lem:eq-or-matching} implies that $|V(H)| \ge h$ if and only if at least one of the following holds true.
    \begin{enumerate}
        \item $G_1$ and $G_2$ both have a matching consisting of $\lceil h/2 \rceil$ edges.
        \item $G_1$ and $G_2$ both have a common star forest subgraph of size exactly $h$.
    \end{enumerate}
    Let $S(h)$ be the set of all star forests with $h$ vertices unique up to isomorphism.
    
    \begin{claim}\label{clm:part-count}
        $|S(h)| = 2^{\OO(\sqrt h)}$, and this can be enumerated in time $2^{\OO(\sqrt h)}$.
    \end{claim}
    \begin{proof}
        A \emph{partition} of an integer $t \ge 0$ is a multiset X containing elements of $[t]$ with $\sum_{x\in X} x = t$. Let $P(h)$ be the set of partitions of $h$. We show that $|S(h)| \le |P(h)|$ by constructing an injective mapping $\phi:S(h) \to P(h)$. Consider $H \in S(h)$ to be a star forest with $h$ vertices, without isolated vertices. Let the sizes of connected components of $H$ be $c_1, c_2, \ldots$, where $c_1 \ge c_2 \ge \ldots$. Note that $\sum_{i}{c_i} = h$, hence $\{c_1, c_2, \ldots\} \in P(h)$. We define $\phi$ such that $H \mapsto \{c_1, c_2, \ldots\}$. Now, let $\phi(H_1) = \phi(H_2)$, where $H_1, H_2 \in S(h)$. Hence, the sizes of connected components of $H_1$ and $H_2$ are same. As all connected components of $H_1$ and $H_2$ are stars, $H_1$ and $H_2$ are isomorphic. Hence, $H_1 = H_2$.

        This proves that $|S(h)| \le |P(h)|$. The partition bound~\cite{andrews1998theory}
        \[
            |P(h)| \le \exp\left(\pi \sqrt{2h/3} + o(h)\right) = 2^{\OO(\sqrt h)}
        \]
        gives $|S(h)| = 2^{\OO(\sqrt h)}$. Enumerating partitions can be done using a simple recursive algorithm in time $\OO(h^2 \cdot |P(h)|) = 2^{\OO(\sqrt h)}$~(Lemma~\ref{lem:part-algo}). Since $S(h)$ exactly contains all partitions in $P(h)$ which do not contain $1$, $S(h)$ can be enumerated in the same time.
    \end{proof}

    We now state the algorithm for {\ourprob} parameterized by $h$. Let $M$ be the graph consisting of $\lceil h/2 \rceil$ copies of an edge (i.e. a matching of size $\lceil h/2 \rceil$). By Lemma~\ref{lem:eq-or-matching}, the instance is a \yes-instance if and only if there exists a common star forest subgraph $H$ of $G_1$ and $G_2$ in $S(h) \cup \{M\}$.

    The algorithm enumerates $S(h)$ (Claim~\ref{clm:part-count}) and for all $H' \in S(h) \cup \{M\}$, it checks whether $H'$ is a subgraph of $G_1$ and a subgraph of $G_2$. Since $H'$ is a forest, Alon, Yuster and Zwick~\cite{alon1995color} show that this can be done in time $2^{\OO(h)} \cdot n$. The algorithm outputs \yes if for some $H' \in S(h) \cup \{M\}$, $H'$ is a common star forest subgraph of $G_1$ and $G_2$; and \no otherwise. 

    The total running time of the algorithm is $2^{\OO(\sqrt{h})} \cdot 2^{\OO(h)} \cdot n = 2^{\OO(h)} \cdot n$.
\end{proof}

\section{Missing Proofs of Section~\ref{subsec:VCmaxdeg}}

\subsection{Proof of Theorem~\ref{thm:wth-vc}}
\begin{proof}
The parameterized problem {\sc Dominating Set} takes as input a graph $G$ and a parameter $k$ and asks if $G$ has a dominating set of size at most $k$. {\sc Dominating Set} is \WTH with $k$ as the parameter~\cite{cygan2015parameterized}. We show a parameterized reduction from {\sc Dominating Set} to \ourprob. Let $(G, k)$ be an instance of {\sc Dominating Set} with $|V(G)| = n$. Without loss of generality, we assume that $G$ has no isolated vertices. Our reduction constructs an instance $(G_1, G_2, h)$ of {\ourprob}. Set $G_1 = G$, $h = n$, and set $G_2$ to be a disjoint collection of $k$ many star graphs, each having $n-1$ leaves. The $i$-th star in $G_2$ has a `centre' $c_i$, which is connected to the `leaves' $v_i^{(1)}, v_i^{(2)}, \ldots, v_i^{(n - 1)}$. The minimum vertex cover of $G_2$ is the set of all centres of the stars, and therefore is of size $k$.
\begin{claim}
    $(G,k)$ is a \yes-instance of {\sc Dominating Set} if and only if $(G_1,G_2,h)$ is a \yes-instance of {\ourprob}.
\end{claim}
\begin{proof}
    We show both implications.
    \textbf{Forward Direction.} Assume that $G_1$ has a dominating set of size at most $k$. Let $D = \{u_1, \ldots, u_q\} \subseteq V$ be a minimum dominating set of $G$ with $q \le k$. Therefore, there exists a matching $M$ exhausting $D$, consisting of edges which are incident on one vertex in $D$ and one vertex outside $D$ (by Lemma~\ref{lem:dom-match}). Let vertex $u_i$ be matched to $w_i \in V(G) \setminus D$ for $i \in [q]$ in $M$. Let $W = \{w_1, \ldots, w_q\}$.

    Let $v \in V(G) \setminus (D \cup W)$. Let $b_v$ be any arbitrary vertex in $N_G(v) \cap D$. For $u_i \in D$, we define $B_i = \{w_i\} \cup \{v \in V(G)\setminus(D \cup W) \mid b_v = u_i\}$. Notice that $1 \le |B_i| \le n-1$ and $B_i \subseteq N_G(u_i)$.

    Therefore, $G_1$ contains as a spanning subgraph, a collection of $q$ stars of total $n$ nodes: the $i$-th $u_i$ and leaves $B_i$. However, since $G_2$ has a collection of $k \ge q$ many stars with $n-1 \ge |B_i|$ leaves, such a collection of stars appears as a subgraph in $G_2$ as well. This proves that $(G_1, G_2, h)$ is a \yes-instance of \ourprob{}.
    
    \textbf{Backward Direction.} Now, let $(G_1, G_2, h)$ be a \yes-instance of \ourprob. Since $h = n = |V(G_1)|$, this means that there is a spanning subgraph $H$ of $G_1$ which is a collection of stars. The centre of the stars in $H$ together constitute a dominating set. Let $r$ be the number of stars in $H$. Since $H$ is also a subgraph of $G_2$, $r \le k$ (as $G_2$ consists $k$ stars). Therefore, $(G,k)$ is a \yes-instance for {\sc Dominating Set}.
\end{proof}
Hence, \ourprob{} is \WTH when the parameter is the size of a minimum vertex cover of one of the graphs.
\end{proof}

\subsection{Proof of Theorem~\ref{thm:FPT-vc}}
\begin{proof}
Let $V_1 = V(G_1)$ and $V_2 = V(G_2)$. Let $H$ be a maximum common star forest subgraph of $G_1$ and $G_2$ (which is unknown to the algorithm). The algorithm runs in two phases: the \emph{preprocessing phase} and the \emph{ILP phase}. The preprocessing phase creates multiple ILP instances, each of which is solved in the ILP phase, which finally outputs the best answer among all instances. 

\textbf{Preprocessing phase.} Let $C_1$ be a minimum vertex cover of $G_1$, with $|C_1| = a \leq k$. Let $I_1 = V_1 \setminus C_1$. Notice that $I_1$ is a maximum independent set of $G_1$. Now, we create $2^a$ twin classes partitioning $I_1$, one for each $X \subseteq C_1$, as follows: define twin class $B^{(1)}_X = \{v \in I_1 \mid N_{G_1}(v) = X\}$. Let $H_1$ (which are unknown to the algorithm) be a subgraph of $G_1$ which is isomorphic to $H$. The number of stars in $H_1$ must be at most $a$, as every star must contain at least one vertex from the vertex cover $C_1$. We classify the stars into two types:
\begin{itemize}
    \item Type I stars: the centre of a Type I star is a vertex from $C_1$.
    \item Type II stars: the centre of a Type II star is a vertex from $I_1$.
\end{itemize}
Note that every leaf of a Type II star belongs to $C_1$.

We `guess' the number of Type I and Type II stars. Let there be $p$ many Type I and $q$ many Type II stars, where $p + q \leq a$. Let $T^{(1)}_1, T^{(1)}_2, \ldots, T^{(1)}_p$ be the Type I stars, and $L^{(1)}_1, L^{(1)}_2, \ldots, L^{(1)}_q$ be the Type II stars. There are $\OO(k^2)$ guesses for $p$ and $q$.

We now `guess' the centres of the Type I stars. Let $c_1, c_2, \ldots, c_p \in C_1$ be the centres of $T^{(1)}_1, T^{(1)}_2, \ldots, T^{(1)}_p$ respectively. There are $\OO(a^p) =\OO(k^k)$ guesses of $c_1, c_2, \ldots, c_p \in C_1$. The variable $\alpha_i$ would denote the value of $|V(T^{(1)}_i)|$, and will be a variable in the ILP we design.

Next we `guess' all vertices of the Type II stars. Notice that two vertices in $B_X^{(1)}$, for any $X \subseteq C_1$ have the same neighbourhood and therefore they are indistinguishable from each other. For $L_i^{(1)}$, $i \in [q]$, let $c_{p+i} \in I_1$ be its centre. We guess $X \subseteq C_1$ such that $c_{p+i} \in B^{(1)}_X$. Note that since all vertices in $B^{(1)}_X$ are indistinguishable from each other, we can arbitrarily consider any vertex to be $c_{p+i}$. Next, we guess the exact leaves of $L^{(1)}_i$, this will be a subset of $X$. For every $i \in q$, there are $2^a$ choices of $X$ and for each $X$, there are $2^{|X|}$ choices of the exact set of leaves of $L^{(1)}_i$. Hence in total, there are at most $\sum \limits_{j \in [a]} \binom{a}{|X|} 2^{|X|} \le 3^k$ guesses. Let $\alpha_{p+i} \in \{2, 3, \ldots, a + 1\}$ be the value of $|V(L^{(1)}_i)|$; this will be a constant in the ILP we design. Also, for every $X \subseteq C_1$, let $a_X = |B_X^{(1)} \setminus \{c_{p+i} \mid i \in [q]\}|$ be the number of nodes in $B_X^{(1)}$ which are not centres of $L^{(1)}_i$ for $i \in [q]$; this will also be a constant in the ILP we design.

For every other remaining vertex $v$ in $C_1$ which is neither guessed to be a centre of $T_i^{(1)}$ nor a leaf of $L_{i'}^{(1)}$, we guess whether it is a leaf of some $T^{(1)}_i$ or if it is not a part of $H_1$. For every $v$ there are $(p+1)$ choices, hence in total there are at most $(p+1)^a = \OO((k+1)^{k})$ many guesses. After this guess, every vertex in $C_1$ is either uniquely guessed to be a centre of $T_i^{(1)}$ or a leaf of $T_i^{(1)}$ or a leaf of $L_{i'}^{(1)}$, or not a part of $H_1$. Let $\beta_i = |T_i^{(1)} \cap C_1|$, this will be a constant in the ILP we design.

For each $X \subseteq C_1$, and $i \in [p]$, we denote by the variable $x_{i, X}$ the value $|T_i^{(1)} \cap B^{(1)}_X|$; this will be a variable in our clause. 

Now, we perform similar guesses for $G_2$. Let $H_2$ (which is unknown to the algorithm) be a subgraph of $G_2$ isomorphic to $H$. Let a minimum vertex cover of $G_2$ be $C_2$, with $|C_2| = b \leq k$; $I_2 = V_2 \setminus C_2$ is a maximum independent set. This gives us twin classes $B^{(2)}_Y = \{v \in I_2 \mid N_{G_2}(v) = Y\}$, for all $Y \subseteq C_2$. Let there be $r$ Type I stars $T^{(2)}_i$'s, and $s$ Type II star $L^{(2)}_i$'s. We ensure $r + s = p + q$. We guess $d_1, d_2, \ldots, d_r \in C_2$ be the centres of $T^{(2)}_1, T^{(2)}_2, \ldots, T^{(2)}_r$ respectively. The variable $\gamma_i$ denotes the value of $|V(T^{(2)}_i)|$. For $L_i^{(2)}$, $i \in [s]$, let $d_{r+i} \in I_2$ be its centre. We guess $Y \subseteq C_2$ such that $d_{r+i} \in B^{(2)}_Y$. We guess the exact leaves of $L^{(2)}_i$ in $C_2$. Let $\gamma_{r+i} \in \{2, 3, \ldots, b + 1\}$ be the value of $|V(L^{(2)}_i)|$, and for every $Y \subseteq C_2$, let $b_Y = |B_Y^{(2)} \setminus \{d_{r+i} \mid i \in [s]\}|$. For each remaining vertex of $C_2$, we guess if it belongs to $T_i^{(2)}$ for some $i \in [r]$ or if it is not in $H_2$. Let $\delta_i = |T_i^{(2)} \cap C_2|$. For each $Y \subseteq C_2$, and $i \in [r]$, we denote the variable $y_{i, Y}$ as the value $|T_i^{(2)} \cap B^{(2)}_Y|$. 

Since $H_1$ and $H_2$ must be isomorphic to each other, there must exist a permutation $\pi : [p+q] \to [r+s]$ such that $\alpha_{i} = \gamma_{\pi(i)}$, for every $i \in [p+q]$. There are $k!$ guesses for $\pi$.

Putting everything together, there are $\OO \left( \left(k^2 \cdot k^k \cdot 3^k \cdot \left(k+1\right)^k \right)^2 \cdot k!\right) = k^{\OO(k)}$ guesses. For each such guess, we formulate an ILP and solve it to determine if there are corresponding subgraphs $H_1$ and $H_2$ and we would like to maximize the total size.

\textbf{ILP phase.} We fix a guess from the preprocessing phase. Our objective is to maximize $|V(H_1)| = \sum \limits_{i \in [p + q]} \alpha_i$ (equivalently $|V(H_2)|$). Since $\alpha_i$ is constant for $p < i \le p+q$, it suffices to maximize $\sum \limits_{i \in [p]} \alpha_i$.

We now look into the constraints which ensure that $H_1$ and $H_2$ are star subgraphs of $G_1$ and $G_2$ respectively. Recall that for every $X \subseteq C_1$, $x_{i, X}$ was the size of $|T_i^{(1)} \cap B_X^{(1)}|$. The sets $T_i^{(1)} \cap B_X^{(1)}$ are pairwise disjoint for two values of $i$, and do not contain any centre of $L^{(1)}_{i'}$, $i' \in [q]$. Therefore, the sum of $x_{i, X}$ over all $i$ must not exceed $a_X$. Similarly, for every $Y \subseteq C_2$, the sum of $y_{i, Y}$ over all $i$ must not exceed $b_Y$.

Now, observe that for $i \in [p]$, for any $X \subseteq C_1 \setminus \{c_i\}$, $T_i^{(1)}$ does not contain any node in $B^{(1)}_X$; and therefore, $x_{i,X} = 0$. Similarly, for $i \in [r]$, for any $Y \subseteq C_2 \setminus \{d_i\}$, $T_i^{(2)}$ does not contain any node in $B^{(2)}_Y$; and therefore, $y_{i,Y} = 0$.

Next, we compute the size of $T_i^{(1)}$ for every $i \in [p]$. It has $\beta_i$ nodes in $C_1$, and it has $\sum \limits_{X \subseteq C_1} x_{i,X}$ nodes in $I_1$. Therefore, the sum of these two terms must be equal to $\alpha_i$. Similarly, for every $i \in [r]$, the sum of $\delta_i$ and $\sum \limits_{Y \subseteq C_2} y_{i,Y}$ must be equal to $\gamma_i$.

Finally, by definition of $\pi$, we must have $\alpha_i = \gamma_{\pi(i)}$ for every $i \in [p+q]$.

Each ILP has $k_1 = a \cdot 2^a + b \cdot 2^b + p + r = 2 ^ {\OO(k)}$ variables, hence each can be solved in time $k_1 ^ {\OO(k_1)} \cdot n^ {\OO(1)} = 2^{2^{\OO(k)}} \cdot n^{\OO(1)}$ time. Since there are $k^{\OO(k)}$, the total running time is $k^{\OO(k)} \cdot 2^{2^{\OO(k)}} \cdot n^{\OO(1)} = 2^{2^{\OO(k)}} \cdot n^{\OO(1)}$. The output of the algorithm is the maximum possible value of $\sum \limits_{i \in [p+q]} \alpha_i$ over all feasible guesses. For completeness we formally state the ILP below.

\begin{align*}
    \max \sum_{i \in [p]} \alpha_i& &\\
    \text{subject to:} & &\\
    &\sum \limits_{i \in [p]} x_{i,X}\le a_X & \forall X \subseteq C_1\\
    &\sum \limits_{i \in [r]} y_{i,Y}\le b_Y &\forall Y \subseteq C_2\\
    &x_{i,X}= 0 & \forall i \in [p], X \subseteq C_1 \setminus \{c_i\}\\
    &y_{i,Y}= 0 & \forall i \in [r], X \subseteq C_2 \setminus \{d_i\}\\
    &\alpha_i= \beta_i + \sum \limits_{X \subseteq C_1} x_{i,X}  &\forall i \in [p]\\
    &\gamma_i= \delta_i + \sum \limits_{Y \subseteq C_2} y_{i,Y} &\forall i \in [r]\\
    &\alpha_i= \gamma_{\pi(i)} &\forall i \in [p+q]
\end{align*}\end{proof}
\section{Missing Proofs of Section~\ref{subsec:other}}

\subsection{Proof of Theorem~\ref{thm:FPT-conn-comp}}

\begin{proof}
    Let the input be graphs $G_1$, $G_2$ and we want to maximize the size of the common star forest subgraph. Let $n = |V(G_1)| + |V(G_2)|$. The \FPT algorithm consists of two phases: the \emph{preprocessing phase} and the \emph{ILP phase}. The preprocessing phase computes the necessary values needed for the ILP, which is then solved in the ILP phase.

    \textbf{Preprocessing phase.} First we enumerate all distinct ordered, connected graphs of at most $k$ nodes. These are named $C_1$, $C_2$, ..., $C_{k_1}$. This can be done by simply enumerating all possible subsets of edges between $k$ vertices. Therefore $k_1 \le 2^{\binom{k}{2}}$, and this process can be done in time $\OO(k_1)$ time.

    Let $p_i$ be the number of copies of $C_i$ appearing in $G_1$, and let $q_i$ be the number of copies of $C_i$ in $G_2$, for all $i \in [k_1]$. This can be done in $\OO(k_1 \cdot k \cdot n)$ time by any simple search algorithm like BFS or DFS. Notice that $(p_1, \ldots, p_{k_1}, q_1, \ldots, q_{k_1})$ completely determine the graphs $G_1$ and $G_2$.

    Next, we enumerate all the set $\{0, 1, \ldots , k-1\}^k$. Let $S_1, S_2, \ldots, S_{k_2}$, for $k_2 = k^k$ be the ordered tuples, where $S_\ell = (s_{\ell}^{(1)}, s_{\ell}^{(2)}, \ldots, s_{\ell}^{(k)})$, $s_{\ell}^{(j)} \in \{0, 1, \ldots, k-1\}$ for all $j \in [k], \ell \in [k_2]$. This enumeration takes time $\OO(k \cdot k_2)$.

    For $i \in [k_1]$, $\ell \in [k_2]$, we define a \emph{realisation of $S_\ell$ into $C_i$} as a subset of edges that induce a star forest in $C_i$ such that for every $j \in [k]$, there are exactly $s_{\ell}^{(j)}$ occurrences of stars of $j$ leaves (and 1 centre). Next we compute the boolean variables $\br_{i,\ell} \in \{\true, \false\}$ for every $i \in [k_1], \ell \in [k_2]$. $\br_{i,\ell}$ is $\true$ if there exists a realisation of $S_\ell$ into $C_i$; $\br_{i,\ell}$ is $\false$ otherwise. Since $|V(C_i)| \le k$, all stars must have at most $k$ leaves, and there cannot be more than $k$ stars. A simple algorithm to compute $\br_{i,\ell}$ can be described as follows. Start by initializing all $\br_{i,\ell}$ to \false, and then iterate over each $i \in [k_1]$: for each subset of edges in $C_i$, check if deleting them leaves $C_i$ with a collection of stars with exactly $s_{\ell}^{(j)}$ occurrences of stars of $j$ leaves, for $j\in[k]$. This entire procedure takes time $\OO(k_1 \cdot 2^{\binom{k}{2}} \cdot k) = \OO(k \cdot k_1^2)$. This concludes the preprocessing phase. 

    \textbf{ILP phase.} To solve \ourprob{}, we need to remove a subset of edges from $G_1$ and $G_2$ such that the subgraph induced by the remaining edges, $H$, is a collection of stars and the number of nodes in $H$ is maximized. Now, fix a component $C$ of $G_1$. We will refer to $H$ as a \emph{solution}. Let $C$ be a copy of $C_i$ for some $i \in [k_1]$. Let the edges in $H \cap C$ induce a star forest where the star with $j \in [k]$ leaves appears $s^{(j)}$ times. There must be some $\ell \in [k_2]$ such that $S_\ell = (s_{\ell}^{(1)}, s_{\ell}^{(2)}, \ldots, s_{\ell}^{(k)}) = (s^{(1)}, s^{(2)}, \ldots, s^{(k)})$, i.e. the edges in $H \cap C$ is a realisation of $S_\ell$ into $C$. Moreover, such a realisation exists if and only if $\br_{i,\ell} = \true$.
    
    We now introduce the variables for the ILP. Among the $p_i$ occurrences of $C_i$ in $G_1$, we denote by the integer variable $x_{i,\ell}$, for $i \in [k_1], \ell \in [k_2]$, the number of components such that a solution $H$ induces a realisation of $S_\ell$ into this connected component. Similarly in $G_2$, let $y_{i,\ell}$, for $i \in [k_1], \ell \in [k_2]$, be the number of occurrences of $C_i$ such that a solution $H$ induces a realisation of $S_\ell$ into that component. These will be the only variables of the ILP that we would formulate.   

    Next, we formulate the constraints of the ILP which would ensure that the variables $x_{i,\ell}$ and $y_{i,\ell}$ correspond to some valid solution. Firstly, for all occurrences $C$ of $C_i$ in $G_1$, the edges of a solution $H$ must induce a realisation of some $S_\ell$, $\ell \in [k_2]$ into $C$. Therefore, for a fixed $i$, the variables $x_{i,\ell}$ must add up to exactly $p_i$; and similarly the variables $y_{i,\ell}$ must add up to exactly $q_i$.

    The next criteria we need to ensure is that to only allow realisations of $S_\ell$ into $C_i$ when $\br_{i,\ell} = \true$. Therefore, for all $i \in [k_1], \ell \in [k_2]$ with $\br_{i,\ell} = \false$, we enforce $x_{i,\ell} = y_{i,\ell} = 0$.

    Finally, for every $j \in [k]$, the edges of a solution $H$ must induce equal number of stars with $j$ leaves in $G_1$ and $G_2$. The number of such stars induced by the edges of $H$ in any occurrence of $C_i$ in $G_1$, with the edges of $H$ inducing a realisation of $S_\ell$ into this component, is $s_{\ell}^{(j)}$. Therefore, the total number of occurrences of such stars induced by the edges $H$ in $G$ is $\sum \limits_{i \in [k_1]} \sum \limits_{\ell \in [k_2]} s_{\ell}^{(j)} \cdot x_{i,j}$. Similarly, the total number of occurrences of such stars induced by the edges $H$ in $G$ is $\sum \limits_{i \in [k_1]} \sum \limits_{\ell \in [k_2]} s_{\ell}^{(j)} \cdot y_{i,j}$. These two quantities must be equal for each $j \in [k]$. This completes the requirements for the variables to correspond to a valid solution.
    
    We want to maximize the number of nodes in such a solution $H$. We argued that the total number of occurrences of stars with $j$ leaves in $H$ is $\sum \limits_{i \in [k_1]} \sum \limits_{\ell \in [k_2]} s_{\ell}^{(j)} \cdot x_{i,j}$. Each such star contains $(j+1)$ nodes. Therefore, the total number of nodes in $H$ is $\sum \limits_{j \in [k]} \sum \limits_{i \in [k_1]} \sum \limits_{\ell \in [k_2]} (j+1) \cdot s_{\ell}^{(j)} \cdot x_{i,j}$. This is the objective that we aim to maximize. 
    
    Therefore, we have an ILP with $k_3 = 2 \cdot k_1 \cdot k_2$ variables. This can be solved in time $k_3^{\OO(k_3)} \cdot n^{\OO(1)} = 2^{2^{\OO(k^2)}} \cdot n^{\OO(1)}$ time~\cite{cygan2015parameterized}. This dominates all other steps from the preprocessing phase. For completeness, we provide the formal ILP formulation below. 

\begin{align*}
    \max \sum \limits_{j \in [k]} \sum \limits_{i \in [k_1]} \sum \limits_{\ell \in [k_2]}& (j+1) \cdot s_{\ell}^{(j)} \cdot x_{i,j} & \\
    \text{subject to:} & & \\
    \sum \limits_{\ell \in [k_2]} x_{i,\ell}&= p_i&\forall i \in [k_1]\\
    \sum \limits_{\ell \in [k_2]} y_{i,\ell}&= q_i&\forall i \in [k_1]\\
    x_{i,\ell}&= 0&\forall i \in [k_1], \ell \in [k_2], \text{ with } \br_{i,\ell} = \false\\
    y_{i,\ell}&= 0&\forall i \in [k_1], \ell \in [k_2], \text{ with } \br_{i,\ell} = \false\\
    \sum \limits_{i \in [k_1]} \sum \limits_{\ell \in [k_2]} s_{\ell}^{(j)} x_{i,j}&=\sum \limits_{i \in [k_1]} \sum \limits_{\ell \in [k_2]} s_{\ell}^{(j)} y_{i,j}& \forall j \in [k]\\
    x_{i,\ell},y_{i,\ell} &\in \mathbb{Z}_{\ge 0}& \forall i \in [k_1], \ell \in [k_2]
\end{align*}    
\end{proof}
\subsection{Proof of Theorem~\ref{thm:woh-deg-const-pw}}
We start by stating below the strongly \NPH problem of {\sc $k$-way Partition}.

\defproblem{$k$-way Partition}{A set of $n$ positive integers $a_1, \ldots, a_{n}$, integers $k$ and $C$, all given in unary.}{Is it possible to partition $[n]$ into $k$ subsets $B_1, \ldots, B_k$ such that for all $j \in [k]$, $\sum\limits_{i \in B_j} a_i = C$.}

The problem of {\sc $k$-way Partition} is not only \NPH but also \WOH when $k$ is the parameter~\cite{dreier2019complexity,heeger2023single,jansen2013bin}. 
 \begin{proof}
    We provide a parameterized reduction from the \WOH problem of {\sc $k$-way Partition} when the parameter is $k$. Let $a_1,\ldots,a_n,k$ and $C$ be the input to the {\sc $k$-way Partition} problem, where $a_i$ is a positive integer. Moreover, we assume without loss of generality that $a_1 \ge a_2 \ge \ldots \ge a_n$.
\textbf{Construction of the reduced instance.} Our reduction constructs an instance of {\ourprob} with graphs $G_1$ and $G_2$. $G_1$ would be a collection of $n$ trees, each of pathwidth at most $4$, and maximum degree $3k+7$; while $G_2$ would contain multiple stars of maximum degree $2k+8$. $G_1$ and $G_2$ would have the exact same number of nodes $n'$, and the target would be equal to the number $n'$.

     Let $M = \sum \limits_{i \in [n]} a_i$, $a = \max \limits_{i \in [n]} a_i$, $D = 2k+8$ and $E = 2k+6$. Note that $kC = M$ (otherwise it is a trivial \no-instance). Now we define the construction of $G_1$. For each $i \in [n]$, include the tree $T_i$ in $G_1$, defined as follows (Figure~\ref{fig:PW-T_i}).
     \begin{itemize}
         \item Construct a node $r_i$.
         \item Construct $D-1$ nodes $h_{i,1},h_{i,2},\ldots,h_{i,D-1}$, each of them adjacent to $r_i$.
         \item Construct $k$ nodes $s_{i,1}^{(0)},s_{i,2}^{(0)}, \ldots, s_{i,k}^{(0)}$, each of them adjacent to $r_i$.
         \item For every $j \in [k]$, construct the following.
         \begin{itemize}
             \item Construct $3(a-2)$ nodes $s_{i,j}^{(1)}, s_{i,j}^{(2)}, \ldots, s_{i,j}^{(a-2)}, y_{i,j}^{(1)}, y_{i,j}^{(2)}, \ldots, y_{i,j}^{(a-2)}$, and $z_{i,j}^{(1)}, z_{i,j}^{(2)}, \ldots, z_{i,j}^{(a-2)}$.
             \item Construct $2a$ nodes $t_{i,j}^{(1)}, t_{i,j}^{(2)}, \ldots, t_{i,j}^{(a)}, u_{i,j}^{(1)}, u_{i,j}^{(2)}, \ldots, u_{i,j}^{(a)}$
             \item For every $\ell \in [a-2]$, construct the following.
             \begin{itemize}
                 \item Construct an edge between $s_{i,j}^{(\ell)}$ and $z_{i,j}^{(\ell)}$.
                 \item Construct an edge between $z_{i,j}^{(\ell)}$ and $y_{i,j}^{(\ell)}$.
                 \item Construct an edge between $y_{i,j}^{(\ell)}$ and $s_{i,j}^{(\ell-1)}$.
                 \item Construct an edge between $t_{i,j}^{(\ell)}$ and $s_{i,j}^{(\ell-1)}$.
             \end{itemize}
             \item Construct an edge between $t_{i,j}^{(a-1)}$ and $s_{i,j}^{(a-2)}$
             \item Construct an edge between $t_{i,j}^{(a)}$ and $s_{i,j}^{(a-2)}$.
             \item For every $\ell \in [a_i]$, construct the following.
             \begin{itemize}
                 \item Construct an edge between $u_{i,j}^{(\ell)}$ and $t_{i,j}^{(\ell)}$.
                 \item Construct $2j+4$ nodes $v_{i,j}^{(\ell,1)}, v_{i,j}^{(\ell,2)}, \ldots, v_{i,j}^{(\ell,2j+4)}$, each of them adjacent to $u_{i,j}^{(\ell)}$. 
             \end{itemize}
             \item For every $\ell \in [a]\setminus[a_i]$, construct the following.
             \begin{itemize}
                 \item Construct an edge between $u_{i,j}^{(\ell)}$ and $t_{i,j}^{(\ell)}$.
                 \item Construct $E$ nodes $v_{i,j}^{(\ell,1)}, v_{i,j}^{(\ell,2)}, \ldots, v_{i,j}^{(\ell,E)}$, each of them adjacent to $u_{i,j}^{(\ell)}$. 
             \end{itemize}
         \end{itemize}
     \end{itemize}
     $G_1$ is simply the collection of the trees $T_1, \ldots, T_n$. The highest degree nodes in $G_1$ are $r_i$ for $i \in [n]$, each having degree $D - 1 + k = 3k + 7$.

     \begin{claim}\label{clm:G_1-pw-5}
         Pathwidth of $G_1$ is at most $4$.
     \end{claim}
     \begin{proof}
        It suffices to show that $T_i$ has pathwidth at most $4$ for every $i \in [n]$. We fix some $i \in [n]$. Consider the graph $T_i - r_i$. It has the following connected components: 
        \begin{itemize}
            \item The singleton connected components $\{h_{i,1}\}, \{h_{i,2}\}, \ldots, \{h_{i,D-1}\}$.
            \item For $j \in [k]$, the connected components containing the vertices
            \begin{align*}
                J_{i,j} = &\left\{s_{i,j}^{(0)}\right\} \cup \left\{s_{i,j}^{(\ell)}, y_{i,j}^{(\ell)}, z_{i,j}^{(\ell)}\mid \ell \in [a-2] \right\} \cup \left\{t_{i,j}^{(\ell)}, u_{i,j}^{(\ell)}\mid \ell \in [a] \right\} \\ 
                &\cup \left\{v_{i,j}^{(\ell,1)}, \ldots, v_{i,j}^{(\ell,2j+4)} \mid \ell \in [a_i]\right\} \cup \left\{v_{i,j}^{(\ell,1)}, \ldots, v_{i,j}^{(\ell,E)} \mid \ell \in [a] \setminus [a_i]\right\}. 
            \end{align*}
        \end{itemize}
        Clearly $\{h_{i,1}\}, \{h_{i,2}\}, \ldots, \{h_{i,D-1}\}$ have pathwidth $0$ each. Now fix $j \in [k]$. We provide a path decomposition of $J_{i,j}$ of width $3$. We define the following `bags' each having $4$ nodes.
        \begin{itemize}
            \item For $\ell \in [a_i - 1]$, $\ell' \in [2j+4]$, define $\mathcal{X}_{i,j}^{(\ell, \ell')} = \left\{s_{i,j}^{(\ell - 1)}, t_{i,j}^{(\ell)}, u_{i,j}^{(\ell)}, v_{i,j}^{(\ell, \ell')}\right\}$.
            \item For $\ell = a_i$, $\ell' \in [2j+4]$: 
            \begin{itemize}
                \item if $a = a_i$, then define $\mathcal{X}_{i,j}^{(\ell, \ell')} =  \left\{s_{i,j}^{(\ell - 2)}, t_{i,j}^{(\ell)}, u_{i,j}^{(\ell)}, v_{i,j}^{(\ell, \ell')}\right\}$,
                \item otherwise, define $\mathcal{X}_{i,j}^{(\ell, \ell')} =  \left\{s_{i,j}^{(\ell - 1)}, t_{i,j}^{(\ell)}, u_{i,j}^{(\ell)}, v_{i,j}^{(\ell, \ell')}\right\}$.
            \end{itemize}
            \item For $\ell \in [a - 1] \setminus [a_i]$, $\ell' \in [E]$, define $\mathcal{Z}_{i,j}^{(\ell, \ell')} = \left\{s_{i,j}^{(\ell - 1)}, t_{i,j}^{(\ell)}, u_{i,j}^{(\ell)}, v_{i,j}^{(\ell, \ell')}\right\}$.
            \item For $\ell' \in [E]$, define $\mathcal{Z}_{i,j}^{(a, \ell')} = \left\{s_{i,j}^{(a-2)}, t_{i,j}^{(a)}, u_{i,j}^{(a)}, v_{i,j}^{(a, \ell')}\right\}$.
            \item For $\ell \in [a - 2]$, define $\mathcal{A}_{i,j}^{(\ell)} = \left\{s_{i,j}^{(\ell - 1)}, y_{i,j}^{(\ell)}, z_{i,j}^{(\ell)}, s_{i,j}^{(\ell)}\right\}$.
        \end{itemize}
        The sequence $\mathcal{X}_{i,j}^{(1, 1)}$, $\ldots$, $\mathcal{X}_{i,j}^{(1, 2j + 4)}$, $\mathcal{A}_{i,j}^{(1)}$, $\mathcal{X}_{i,j}^{(2, 1)}$, $\ldots$, $\mathcal{X}_{i,j}^{(2, 2j + 4)}$, $\mathcal{A}_{i,j}^{(2)}$, $\mathcal{X}_{i,j}^{(3, 1)}$, $\ldots$, $\mathcal{X}_{i,j}^{(a_i, 2j + 4)}$, $\mathcal{A}_{i,j}^{(a_i)}$, $\mathcal{Z}_{i,j}^{(a_i+1, 1)}$, $\ldots$, $\mathcal{Z}_{i,j}^{(a_i + 1,E)}$, $\mathcal{A}_{i,j}^{(a_i + 1)}$, $\mathcal{Z}_{i,j}^{(a_i+2, 1)}$, $\ldots$, $\mathcal{Z}_{i,j}^{(a - 2,E)}$, $\mathcal{A}_{i,j}^{(a - 2)}$, $\mathcal{Z}_{i,j}^{(a - 1, 1)}$, $\ldots$, $\mathcal{Z}_{i,j}^{(a - 1,E)}$, $\mathcal{Z}_{i,j}^{(a, 1)}$, $\ldots$, $\mathcal{Z}_{i,j}^{(a,E)}$ gives a path decomposition of $G_1[J_{i,j}]$ of width $3$. 

        Therefore $T_i - r_i$ has a path decomposition of width at most $3$, achieved by concatenating the path decompositions of each its connected components. Hence $T_i$ has path decomposition of width at most $4$, achieved by adding $r_i$ to all bags of the path decompositions of $T_i$.
     \end{proof}
    Moreover, the number of nodes $T_i$ contains is $1+(D-1)+k+3k(a-2)+2ka+k(a-a_i)E+a_i\cdot\sum\limits_{j\in[k]}(2j+4) = (2ak^2+11ak-3k+8)-(k^2+k)a_i$. Hence, $G_1$ contains a total of $n'=n(2ak^2+11ak-3k+8)-(k^2+k)M$ nodes, a polynomial in the input size of the instance.
      
    We now define the construction of $G_2$. This will contain the following star graphs (Figure~\ref{fig:PW-G_2}).
    \begin{figure} [!htbp]
        \begin{minipage}[t]{0.57\textwidth} 
            \includegraphics[height=0.92\textheight]{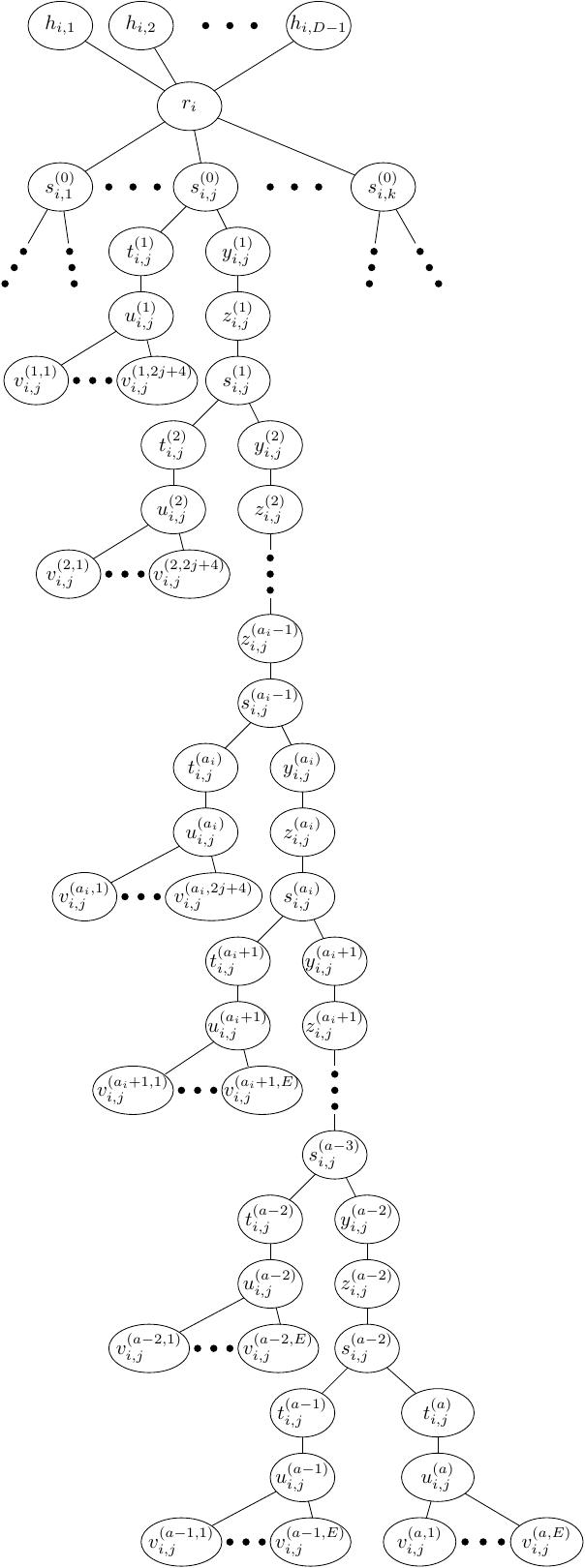}
            \subcaption{Construction of $T_i$.} \label{fig:PW-T_i}
        \end{minipage} \hfill \hfill 
        \begin{minipage}[t]{0.42\textwidth}
            \includegraphics[height=0.83\textheight]{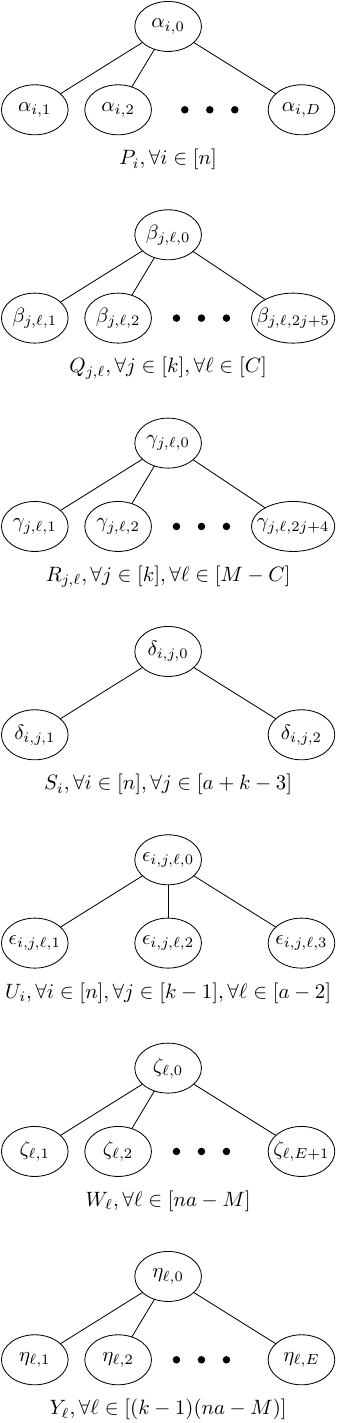}
            \subcaption{Construction of $G_2$.} \label{fig:PW-G_2}
        \end{minipage}
        \caption{Constructions of $T_i$ and $G_2$.}
    \end{figure}
 
    \begin{itemize}
        \item For every $i \in [n]$, construct a star with centre $\alpha_{i,0}$ and $D$ leaves $\alpha_{i,1}, \ldots, \alpha_{i,D}$, each connected to $\alpha_{i,0}$ with an edge; call this star $P_i$.
        \item For every $j \in [k]$ and every $\ell \in [C]$, construct a star with centre $\beta_{j,\ell,0}$ and $2j+5$ leaves $\beta_{j,\ell,1}, \ldots, \beta_{j,\ell,2j+5}$, each connected to $\beta_{j,\ell,0}$ with an edge; call this star $Q_{j,\ell}$.
        \item For every $j \in [k]$ and every $\ell \in [M-C]$, construct a star with centre $\gamma_{j,\ell,0}$ and $2j+4$ leaves $\gamma_{j,\ell,1}, \ldots, \gamma_{j,\ell,2j+4}$, each connected to $\gamma_{j,\ell,0}$ with an edge; call this star $R_{j,\ell}$.
        \item For every $i \in [n], j \in [a+k-3]$, construct a star with centre $\delta_{i,j,0}$ and $2$ leaves $\delta_{i,j,1}$ and $\delta_{i,j,2}$, each connected to $\delta_{i,j,0}$ with an edge; call this star $S_{i,j}$.
        \item For every $i \in [n], j \in [k-1]$ and every $\ell \in [a-2]$, construct a star with centre $\epsilon_{i,j,\ell,0}$ and $3$ leaves $\epsilon_{i,j,\ell,1}, \epsilon_{i,j,\ell,2}$ and $\epsilon_{i,j,\ell,3}$, each connected to $\epsilon_{i,j,\ell,0}$ with an edge; call this star $U_{i,j,\ell}$.
        \item For every $\ell \in [na-M]$, construct a star with centre $\zeta_{\ell,0}$ and $E+1$ leaves $\zeta_{\ell,1}, \zeta_{\ell,2}, \ldots,\zeta_{\ell,E+1}$, each connected to $\zeta_{\ell,0}$ with an edge; call this star $W_\ell$.
        \item For every $\ell \in [(k-1)(na-M)]$, construct a star with centre $\eta_{\ell,0}$ and $E$ leaves $\eta_{\ell,1}, \eta_{\ell,2}, \ldots, \eta_{\ell,E}$, each connected to $\eta_{\ell,0}$ with an edge; call this star $Y_\ell$.
    \end{itemize}
    The total number of nodes in $G_2$ is $n(D+1)+\sum\limits_{j\in[k]}C(2j+6) + \sum\limits_{j\in[k]}(M-C)(2j+5)+3n(a+k-3)+4n(k-1)(a-2)+(na-M)(E+2)+(k-1)(na-M)(E+1) = n(2ak^2+11ak-3k+8)-(k^2+k)M = n'$. Moreover as $G_2$ is a collection of stars, it has a pathwidth of $1$; and the highest degree nodes of $G_2$ are $\alpha_{i,0}$ for $i \in [n]$, each having a degree of $D = 2k + 8$. We set the target to be $n'$. This reduction can clearly be done in polynomial time. Since $G_2$ is already a star-forest of size $n'$, we use Observation~\ref{obs:G2-sub-G1} and decide if $G_2$ is a subgraph of $G_1$.
    
    \textbf{Proof of correctness: forward direction.} Assume that the {\sc $k$-way Partition} instance is a \yes-instance. Therefore, there exists a partition $B_1, \ldots, B_k$ of $[n]$ such that for every $j \in [k]$, $\sum \limits_{i \in B_j} a_i = C$. We show that this would imply that $G_2$ is a subgraph of $G_1$; and hence by Observation~\ref{obs:G2-sub-G1}, the reduced instance would be a \yes-instance. 
    We define $b_i \in [k]$, such that $i \in B_{b_i}$. Let $p_i = \sum\limits_{i' \in B_{b_i}, i' < i} a_i$ be the sum of all items in bin $B_{b_i}$ with index less than $i$. For every $j \in [k] \setminus \{b_i\}$, let $q_{i,j} = \sum\limits_{i' \in [n] \setminus B_j, i' < i} a_i$ be the sum of all items outside bin $B_j$ with index less than $a_i$. Also let $e_i = \sum\limits_{i' < i}(a - a_i)$ and $f_{i,j} = (k-1)e_i + (a-a_i)\cdot|[j] \setminus\{b_i\}|$, for every $j \in [k] \setminus \{b_i\}$. The terms $e_i$ and $f_i$ would help us later in the reduction to keep track of the number of $W_\ell$’s and $Y_\ell$’s embedded into $G_1$.

    We provide an exact embedding of stars in $G_2$ into each of the trees $T_i$ of $G_1$ (depicted in Figure~\ref{fig:T_i-PW-in-G_2}). The star $P_i$ gets embedded into $T_i$ with $r_i$ as the centre and apart from $h_{i,1}, \ldots, h_{i,D-1}$, it includes $s_{i,b_i}^{(0)}$. For all $\ell \in [a-2]$, the nodes $y_{i,b_i}^{(\ell)}$, $z_{i,b_i}^{(\ell)}$, $s_{i,b_i}^{(\ell)}$ embed star $S_{i,\ell}$. For all $j \in [k] \setminus \{b_i\}$, the nodes $s_{i,j}^{(0)}$, $t_{i,j}^{(1)}$, and $y_{i,j}^{(1)}$ embed star $S_{i,j'}$, for every $j' \in [a+k-3] \setminus [a-2]$. For all $j \in [k] \setminus \{b_i\}$, for all $\ell \in [a-3]$, the nodes $s_{i,j}^{(\ell)}$, $z_{i,j}^{(\ell)}$, $t_{i,j}^{(\ell+1)}$ and $y_{i,j}^{(\ell+1)}$ embed star $U_{i,j',\ell}$, for every $j' \in [k-1]$. For all $j \in [k] \setminus \{b_i\}$, $s_{i,j}^{(a-2)}$, $z_{i,j}^{(a-2)}$, $t_{i,j}^{(a-1)}$ and $t_{i,j}^{(a)}$ embed star $U_{i,j',a-2}$, for every $j' \in [k-1]$. For every $\ell \in [a_i]$, $Q_{b_i,\ell + p_i}$ gets embedded into $T_i$ with $u_{i,b_i}^{(\ell)}$ as its centre and $v_{i,b_i}^{(\ell,1)}, v_{i,b_i}^{(\ell,2)}, \ldots, v_{i,b_i}^{(\ell,2b_i + 4)}$ and $t_{i,b_i}^{(\ell)}$ as leaves. Due to the definition of $p_i$, the lowest indexed item $i^*$ in $B_{b_i}$ with $i^* > i$ embeds $Q_{b_{i^*}, \ell + p_{i^*}} = Q_{b_i, \ell + (p_i + a_i)}$ for every $\ell \in [a_{i^*}]$. Hence, every $Q_{j',\ell'}$ gets embedded into exactly one $T_{i'}$. For all $j \in [k] \setminus \{b_i\}$ and for all $\ell \in [a_i]$, $R_{j,\ell + q_{i,j}}$ gets embedded into $T_i$ with $u_{i,j}^{(\ell)}$ as its centre and $v_{i,j}^{(\ell,1)}, v_{i,j}^{(\ell,2)}, \ldots, v_{i,j}^{(\ell,2j + 4)}$ as leaves. Due to the definition of $q_{i,j}$, the lowest indexed item $i^*$ in $[n] \setminus B_j$ with $i^* > i$ embeds $R_{b_{i^*}, \ell + q_{i^*,j}} = R_{b_i, \ell + (q_{i,j} + a_i)}$ for every $\ell \in [a_{i^*}]$. Hence, every $R_{j',\ell'}$ gets embedded into exactly one $T_{i'}$. For all $\ell \in [a-a_i]$, $W_{\ell+e_i}$ gets embedded to $T_i$ with $u_{i,b_i}^{(\ell+a_i)}$ as its centre and $v_{i,b_i}^{(\ell+a_i,1)}, v_{i,b_i}^{(\ell+a_i,2)}, \ldots, v_{i,b_i}^{(\ell+a_i,E)}$ and $t_{i,b_i}^{(\ell+a_i)}$ as leaves. Due to the definition of $e_i$, the lowest indexed item $i^*$ in $B_{b_i}$ with $i^* > i$ embeds $W_{\ell + e_{i^*}} = W_{\ell + (e_i + a - a_i)}$ for every $\ell \in [a_{i^*}]$. Hence, every $W_{\ell'}$ gets embedded into exactly one $T_{i'}$. Finally, for all $j \in [k] \setminus \{b_i\}$, and for all $\ell \in [a-a_i]$, $Y_{\ell+f_{i,j}}$ gets embedded to $T_i$ with $u_{i,j}^{(\ell+a_i)}$ as its centre and $v_{i,j}^{(\ell+a_i,1)}, v_{i,j}^{(\ell+a_i,2)}, \ldots, v_{i,j}^{(\ell+a_i,E)}$ as leaves. Due to the definition of $f_{i,j}$, the lowest indexed item $i^*$ in $[n] \setminus B_j$ with $i^* > i$ embeds $Y_{\ell + f_{i^*,j}} = Y_{\ell + (f_{i,j} + (k-1)(a-a_i))}$ for every $\ell \in [a_{i^*}]$. Hence, every $Y_{\ell'}$ gets embedded into exactly one $T_{i'}$. This completes the proof of correctness in the forward direction. 
    \begin{figure}[!htbp]
        \includegraphics[width=\textwidth]{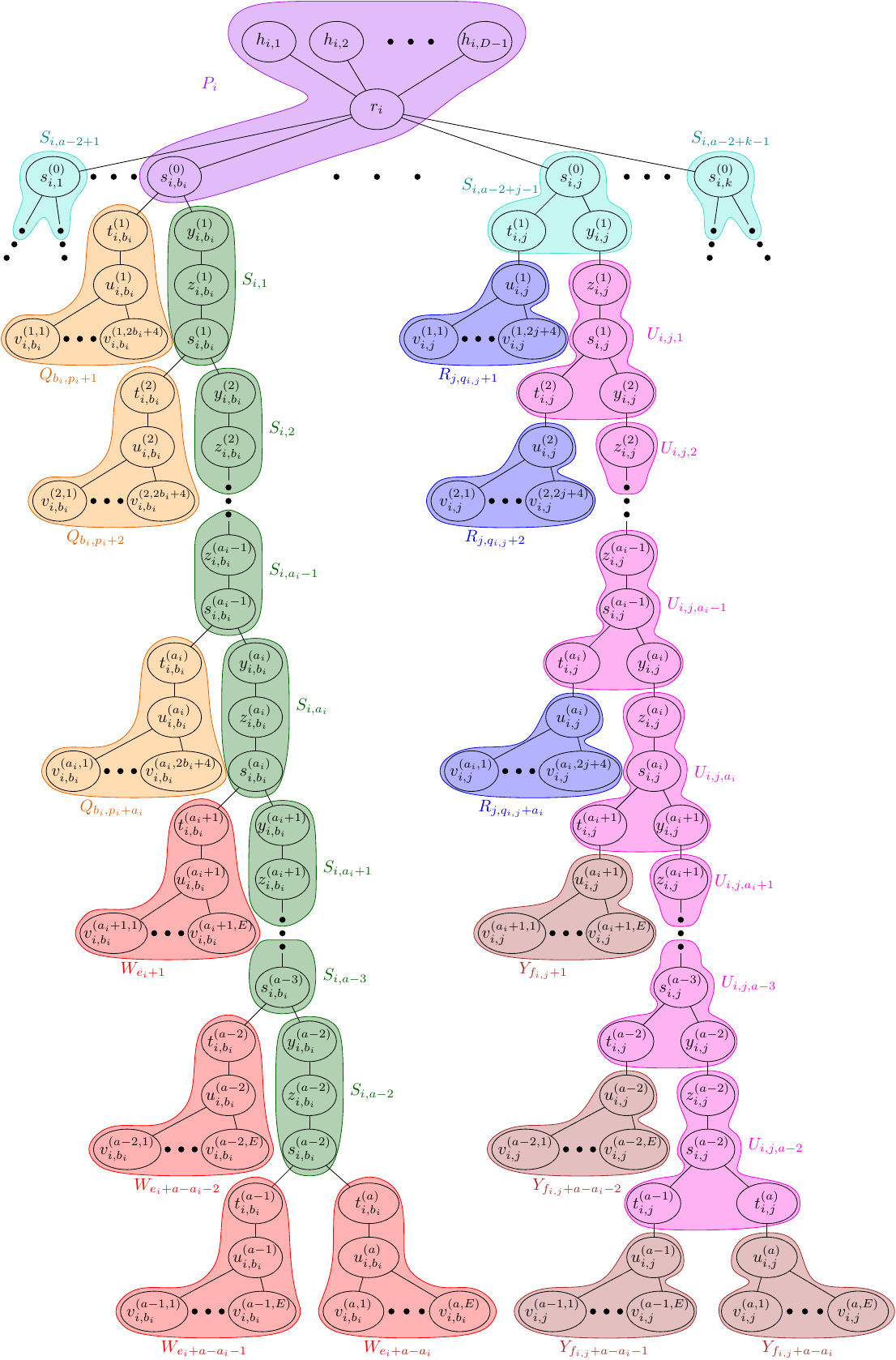}
        \caption{Embedding of stars of $G_2$ into $T_i$ for branch $b_i$ and branches $j \ne b_i$.} \label{fig:T_i-PW-in-G_2}
    \end{figure} 

    \textbf{Proof of correctness: backward direction.} Assume that the reduced instance is a \yes-instance, i.e. $G_2$ is a subgraph of $G_1$ (by Observation~\ref{obs:G2-sub-G1}). We show that the {\sc $k$-way Partition} instance is also a \yes-instance by proving a series of claims. Intuitively, these claims together prove that there exists an embedding identical to what we had in the forward direction, i.e., like Figure~\ref{fig:T_i-PW-in-G_2}.

    \begin{claim}\label{clm:pw-P-good}
        There exists an embedding of $G_2$ in $G_1$ such that for every $i \in [n]$, $P_i$ gets embedded into $T_i$ with $r_i$ as the centre and apart from $h_{i,1}, \ldots, h_{i,D-1}$, it includes $s_{i,b_i}^{(0)}$ for some $b_i \in [k]$. We call such an embedding P-consistent.
    \end{claim}
    
    \begin{proof}
        As $G_2$ is a subgraph of $G_1$, $\forall i \in [n], P_i$ has an embedding in $G_1$. Let $\exists i^* \in [n]$ such that the embedding of $\alpha_{i^*,0}$ in $G_1$ is not on $r_i$, for any $i \in [n]$. In that case, let $w^*$ be the embedding of $\alpha_{i^*,0}$ in $G_1$.  Notice that all vertices in $G_1$ other than $r_i$'s have degree at most $a_1 + 1 < D$. Therefore, $\degree(w^*) < D = \degree(\alpha_{i^*,0})$. So, $\exists j \in [D]$ such that embedding of $\alpha_{i^*,j}$ is not a neighbour of $w^*$. But, then the edge $\{\alpha_{i^*,0}, \alpha_{i^*,j}\}$ is not present in the embedding of $P_i^*$ in $G_1$, contradicting the definition of the embedding.
        
        This proves that $\forall i \in [n] \ \exists j \in[n]$, embedding of $\alpha_{i,0}$ in $G_1$ is $r_j$. Without loss of generality, we can assume that $\alpha_{i,0}$ is embedded in $r_i$. Now, let us assume that $\exists i^* \in [n], j \in [D-1]$ such that $h_{i^*,j}$ is not included in the embedding of $P_{i^*}$. As $|V(G_2)| = |V(G_1)|, \ \exists w^* \in V(G_2 \setminus P_{i^*})$ such that $w^*$ is embedded to $h_{i^*,j}$. As no vertex in $G_2 \setminus P_{i^*}$ is isolated, $\exists x^* \in V(G_2 \setminus P_{i^*})$ such that $\{x^*, w^*\}$ is an edge in $G_2$. Note that $x^*$ is not embedded to $r_i$, hence the edge $\{x^*, w^*\}$ is not present in the embedding of $G_2$ in $G_1$, a contradiction.
        
        This proves that there exists an embedding where $\forall i \in [n]$, $P_i$ gets embedded into $T_i$ with $r_i$ as the centre and it includes $h_{i,1}, \ldots, h_{i,D-1}$. As $\degree(\alpha_{i,0}) = D$, embedding of $P_i$ into $T_i$ includes $D$ neighbours of $r_i$. $D-1$ of these neighbours are $h_{i,1}, \ldots, h_{i,D-1}$. Hence, exactly one of the remaining neighbours are included, among $s_{i,1}^{(0)}, \ldots s_{i,k}^{(0)}$; we denote it by $s_{i,b_i}^{(0)}$.
    \end{proof}
    \begin{claim}\label{clm:pw-PS*-good}
        There exists a P-consistent embedding of $G_2$ in $G_1$ such that for all $\ell \in [a-2]$, the nodes $z_{i,b_i}^{(\ell)}$, $y_{i,b_i}^{(\ell)}$ and $s_{i,b_i}^{(\ell)}$ embeds star $S_{i,\ell}$. We call such an embedding PS*-consistent. 
    \end{claim}
    \begin{proof}
        For a P-consistent embedding $\phi$, let $\ell \in [a-2]$ be smallest index such that, $\exists i \in [n]$ such that $S_{i,\ell}$ is not embedded by the nodes $z_{i,b_i}^{(\ell)}$, $y_{i,b_i}^{(\ell)}$ and $s_{i,b_i}^{(\ell)}$. We call $\ell$ the first violator of $\phi$, denoted $\mathrm{viol}(\phi)$. If there are no PS*-consistent embedding, then every P-consistent embedding has a first violator. Consider an embedding $\phi^* \in \underset{\phi \text{ is P-consistent}}{\argmax}\left\{\mathrm{viol}(\phi)\right\}$ be any embedding with the highest index of the first violator. Let $\ell^* = \mathrm{viol}(\phi^*)$. As $\phi^*$ is a P-consistent embedding, $r_i$ embeds $\alpha_{i,0}$ $\forall i \in [n]$ (due to Claim~\ref{clm:pw-P-good}). Also by assumption, for all $\ell \in [\ell^*-1]$, for all $i \in [n]$, the nodes $z_{i,b_i}^{(\ell)}$, $y_{i,b_i}^{(\ell)}$ and $s_{i,b_i}^{(\ell)}$ embed star $S_{i,\ell}$. Now, as $G_2$ is a subgraph of $G_1$, the vertex $y_{i,b_i}^{(\ell^*)}$ must embed some vertex $w$ from $G_2$. The only neighbour of $y_{i,b_i}^{(\ell^*)}$ in $G_1[V(G_1) \setminus (\{r_i \mid i \in [n]\} \cup \{s_{i,b_i}^{(\ell)} \mid i \in [n], \ell \in [\ell^*]\})]$ is $z_{i,b_i}^{(\ell)}$. As all stars in $G_2$ have centres with degree $\ge 2$, $y_{i,b_i}^{(\ell^*)}$ cannot embed a centre. Hence, $z_{i,b_i}^{(\ell^*)}$ embeds the centre of a star in $G_2$, which has $w$ as its leaf. Now, $\degree(z_{i,b_i}^{(\ell^*)})=2$, which is the smallest degree of any centres among the stars in $G_2$. Hence, $z_{i,b_i}^{(\ell^*)}$ embeds $\delta_{i',\ell',0}$, from some $i'\in[n], \ell'\in[a+k-3]$. As all $S_{i',l'}$ are isomorphic, we assume without loss of generality that $z_{i,b_i}^{(\ell^*)}$ embeds $\delta_{i,\ell^*,0}$. This implies $\mathrm{viol}(\phi^*) > \ell^*$, which is a contradiction. 
    \end{proof}
    \begin{claim}\label{clm:pw-PS-good}
        There exists a PS*-consistent embedding of $G_2$ in $G_1$ such that for all $j \in [k] \setminus \{b_i\}$, the nodes $s_{i,j}^{(0)}$, $t_{i,j}^{(1)}$, $y_{i,j}^{(1)}$ embeds star $S_{i,j'}$ for some $j \in [a+k-3] \setminus [a-2]$. We call such an embedding PS-consistent. 
    \end{claim}
    \begin{proof}
        Consider the graph $G_1[V(G_1) \setminus (\{r_i \mid i \in [n]\})]$. For every $j \in [k] \setminus \{b_i\}$, the vertex $s_{i,j}^{(0)}$ has degree $2$. The neighbours of $s_{i,j}^{(0)}$ are $t_{i,j}^{(1)}$ and $y_{i,j}^{(1)}$, both of which have degree $2$ each. Now, as $G_2$ is a subgraph of $G_1$, $s_{i,j}^{(0)}$ needs to embed a vertex $w$ from $G_2$. Now, let $x$ be the centre of the star in $G_2$ containing $w$. $x$ is embedded by a vertex among $s_{i,j}^{(0)}$, $t_{i,j}^{(1)}$ and $y_{i,j}^{(1)}$. If $\degree(x) \ge 3$, then $\exists w^* \in V(G_2)$ such that $\{w^*,x\}$ is an edge in $G_2$, which is not embedded in $G_1$, a contradiction. Hence, $\degree(x) = 2$, and $x=\delta_{i',\ell'}$ for some $i' \in [n], \ell' \in [a+k-3] \setminus [a-2]$. This proves that $s_{i,j}^{(0)}$ embeds a node from some $S_{i',j'}$ in $G_2$.
        
        Note that there are $n(k-1)$ such nodes $s_{i,j}^{(0)}$, which exhaust the $n(k-1)$ stars $S_{i',\ell'}$. As all such $S_{i',\ell'}$ are isomorphic, we can say without loss of generality that $s_{i,j}^{(0)}$ embeds a node from $S_{i,\ell'}$. Note that in $G_1[V(G_1) \setminus (\{r_i \mid i \in [n]\})]$, $\degree(y_{i,j}^{(1)})=2$, and both its neighbours have degree $2$ as well. Following the same argument as the previous paragraph, we see that $y_{i,j}^{(1)}$ embeds some vertex from $S_{i',\ell'}$ as well. Now, as all $s_{i,j}^{(0)}$ exhaust the $S_{i',\ell'}$, $s_{i,j}^{(0)}$ and $y_{i,j}^{(1)}$ embed nodes from the same star in $G_2$. This implies that $\delta_{i,\ell',0}$ is embedded in $s_{i,j}^{(0)}$ or $y_{i,j}^{(1)}$. If it is embedded in $s_{i,j}^{(0)}$, $y_{i,j}^{(1)}$ and $t_{i,j}^{(1)}$ embed the leaves of $S_{i,\ell'}$. If it is embedded in $y_{i,j}^{(1)}$, $s_{i,j}^{(0)}$ and $z_{i,j}^{(1)}$ embed the leaves of $S_{i,\ell'}$.
        
        For a PS*-consistent embedding $\phi$, let $i \in [n]$ be the smallest index such that $\exists j \in [k] \setminus \{b_i\}$, $S_{i,\ell'}$ is embedded by the nodes $y_{i,j}^{(1)}$, $s_{i,j}^{(0)}$ and $z_{i,j}^{(1)}$. We call $i$ the first violator of $\phi$, denoted $\mathrm{viol}(\phi)$. If there is no PS-consistent embedding, then every PS*-consistent embedding has a first violator. Consider the embedding $\phi^* \in \underset{\phi \text{ is PS*-consistent}}{\argmax}\left\{\mathrm{viol}(\phi)\right\}$ be any embedding with the highest index of the first violator. Let $i^* = \mathrm{viol}(\phi^*)$. By assumption, for all $i \in [i^*-1]$, $\forall j \in [k] \setminus \{b_i\}$ the nodes $s_{i,j}^{(0)}$, $y_{i,j}^{(1)}$ and $t_{i,j}^{(1)}$ embed $S_{i,\ell'}$ for some $\ell' \in [a+k-3] \setminus [a-2]$. Now, let $y_{i^*,j}^{(1)}$, $s_{i^*,j}^{(0)}$ and $z_{i^*,j}^{(1)}$ embed the star $S_{i^*,\ell}$. Now, consider the node $s_{i^*,j}^{(1)}$. As all $S_{i,\ell'}$ are exhausted by $s_{i,j}^{(0)}$, the node $s_{i^*,j}^{(1)}$ embeds a node $w$ from a star in $G_2$ which has at least $3$ leaves. Let the centre of that star be $x$. But, only $s_{i^*,j}^{(1)}$ has only $2$ neighbours which have not already been embedded, and both of them have degree $2$. Hence, if any of these embed $x$, $\exists w^* \in V(G_2)$ such that $\{w^*,x\}$ is an edge in $G_2$ which is not embedded in $G_1$, which is a contradiction. Hence, $s_{i^*,j}^{(1)}$ cannot embed any vertex. But this contradicts the fact that $G_2$ is a subgraph of $G_1$, and $|V(G_2)| = |V(G_1)|$.
    \end{proof}
    \begin{claim}\label{clm:pw-PSU-good}
        There exists a PS-consistent embedding of $G_2$ in $G_1$ such that for all $j \in [k] \setminus \{b_i\}$, for all $\ell \in [a-3]$, the nodes $s_{i,j}^{(\ell)}$, $z_{i,j}^{(\ell)}$, $t_{i,j}^{(\ell+1)}$, $y_{i,j}^{(\ell+1)}$ embed star $U_{i,j,\ell}$, and $j \in [k] \setminus \{b_i\}$, the nodes $s_{i,j}^{(a-2)}$, $z_{i,j}^{(a-2)}$, $t_{i,j}^{(a-1)}$, $t_{i,j}^{(a)}$ embed star $U_{i,j,a-2}$. We call such an embedding PSU-consistent. 
    \end{claim}
    \begin{proof}
        Let a PS-consistent embedding be called PSU$^{(d)}$-consistent if $\forall j \in [k] \setminus \{b_i\}, \forall \ell \in [d]$, the nodes $s_{i,j}^{(\ell)}$, $z_{i,j}^{(\ell)}$, $t_{i,j}^{(\ell+1)}$, $y_{i,j}^{(\ell+1)}$ embed star $U_{i,j,\ell}$. We prove the claim via induction on $d$.
        
        \textbf{Base Case:} All PS-consistent embeddings are PSU$^{(0)}$-consistent embeddings by default.

        \textbf{Inductive Step:} Consider a PSU$^{(\ell)}$-consistent embedding. The vertex $y_{i,j}^{(\ell)}$ embeds a leaf of star $U_{i,j,\ell}$. Now, consider the vertex $z_{i,j}^{(\ell)}$. If this vertex embeds a centre of some star in $G_2$, it needs at least $3$ neighbours where the leaves of that star will be embedded. But, $\degree(z_{i,j}^{(\ell)}) = 2$. Hence, $z_{i,j}^{(\ell)}$ embeds a leaf of some star in $G_2$. The centre of this star is embedded in $s_{i,j}^{(\ell)}$. Now, $\degree(s_{i,j}^{(\ell)}) = 3$, hence $s_{i,j}^{(\ell)}$ can only embed $\epsilon_{i',j',\ell',0}$ for some $i', j', \ell'$. As all $U_{i',j',\ell'}$ are isomorphic, we can say without loss of generality that $s_{i,j}^{(\ell)}$ embeds $\epsilon_{i,j,\ell,0}$. Now, $s_{i,j}^{(0)}$ has $3$ neighbours, which embed the leaves $\epsilon_{i,j,\ell,1}$, $\epsilon_{i,j,\ell,2}$, $\epsilon_{i,j,\ell,3}$, which embed $z_{i,j}^{(\ell)}$, $t_{i,j}^{(\ell+1)}$, $y_{i,j}^{(\ell+1)}$. Following this argument for all $i$ and $j$, we find a PSU$^{(l+1)}$-consistent embedding.

        This proves the existence of PSU$^{(a-3)}$-consistent embeddings. Now, consider the vertex $z_{i,j}^{(a-2)}$. If this vertex embeds a centre of some star in $G_2$, it needs at least $3$ neighbours where the leaves of that star will be embedded. But, $\degree(z_{i,j}^{(a-2)}) = 2$. Hence, $z_{i,j}^{(a-2)}$ embeds a leaf of some star in $G_2$. The centre of this star is embedded in $s_{i,j}^{(a-2)}$. Now, $\degree(s_{i,j}^{(a-2)}) = 3$, hence $s_{i,j}^{(a-2)}$ can only embed $\epsilon_{i',j',\ell',0}$ for some $i', j', \ell'$. As all $U_{i',j',\ell'}$ are isomorphic, we can say without loss of generality that $s_{i,j}^{(a-2)}$ embeds $\epsilon_{i,j,a-2,0}$. Now, $s_{i,j}^{(0)}$ has $3$ neighbours, which embed the leaves $\epsilon_{i,j,a-2,1}$, $\epsilon_{i,j,a-2,2}$, $\epsilon_{i,j,a-2,3}$, which embed $z_{i,j}^{(a-2)}$, $t_{i,j}^{(a-1)}$, $t_{i,j}^{(a)}$. Following this argument for all $i$ and $j$, we find a PSU-consistent embedding.
    \end{proof}
    \begin{claim}\label{clm:pw-PQSU-good}
        There exists a PSU-consistent embedding of $G_2$ in $G_1$ such that for all $\ell \in [a_i]$, $Q_{b_i,\ell+p_i}$ gets embedded into $T_i$ with $u_{i,b_i}^{(\ell)}$ as its centre and $v_{i,b_i}^{(\ell,1)}, v_{i,b_i}^{(\ell,2)}, \ldots, v_{i,b_i}^{(\ell,2b_i + 4)}$ and $t_{i,b_i}^{(\ell)}$ as leaves, where $p_i = \sum\limits_{i' \in [n], i' < i, b_i = b_{i'}} a_i$. We call such an embedding PQSU-consistent.
    \end{claim}
    \begin{proof}
        For any PSU-consistent embedding, all vertices in $X = \{h_{i,j} \mid i \in [n], j \in [D-1]\} \cup \{r_i \mid i\in[n]\} \cup \{s_{i,j}^{(\ell)} \mid i \in [n], j \in [k] \setminus \{b_i\}, \ell \in [a-2] \cup \{0\}\} \cup \{y_{i,j}^{(\ell)} \mid i \in [n], j \in [k] \setminus \{b_i\}, \ell \in [a-2] \} \cup \{z_{i,j}^{(\ell)} \mid i \in [n], j \in [k] \setminus \{b_i\}, \ell \in [a-2] \} \cup \{t_{i,j}^{(\ell)} \mid i \in [n], j \in [k] \setminus\{b_i\},\ell \in [a_i]\}$ have embedded all $P_i$, $S_{i,j}$, and $U_{i,j,\ell}$. Consider the remainder graph $G_1^R = G_1[V(G_1) \setminus X]$. This graph is a star forest, where all stars are centred at some $u_{i,j}^{(\ell)}$ . Now, as $G_2$ is a subgraph of $G_1$, any star with $q$ leaves in $\{Q_{j,\ell}\} \cup \{R_{j,\ell}\} \cup \{W_\ell\} \cup \{Y_\ell\}$ must be embedded to a star in $G_1^R$ with $q$ leaves. The star $Q_{b_i,\ell+p_i}$ has $2b_i+5$ leaves, so it must be embedded to a star in $G_1^R$ with $2b_i+5$ leaves. But, all stars in $G_1^R$ with $2b_i+5$ leaves are centred in $u_{i',b_i}^{(\ell)}$ with $v_{i',b_i}^{(\ell,1)}, v_{i',b_i}^{(\ell,2)}, \ldots, v_{i',b_i}^{(\ell,2b_i + 4)}$ and $t_{i',b_i}^{(\ell)}$ for some $\ell \in [a_i]$, where $b_{i'} = b_i$. As all of these graphs are isomorphic, we can consider without loss of generality that $Q_{b_i,\ell+p_i}$ is embedded in $G_1$ with $\beta_{b_i,\ell+p_i,0}$ being embedded to $u_{i,b_i}^{(\ell)}$. 
    \end{proof}
    \begin{claim}
        For any PQSU-consistent embedding of $G_2$ in $G_1$, for every $j \in [k]$, let $B_j = \{i \in [n] \mid b_i=j\}$. Then for all $j \in [k]$, $\sum\limits_{i \in B_j}a_i = C$.
    \end{claim}
    \begin{proof}
        Consider a PQSU-consistent embedding of $G_2$ in $G_1$. Let $Q_{j,\ell}$ be embedded into some $T_i$. This implies $b_i=j$, as $T_i$ only embeds $Q_{b_i,\ell+p_i}$ for $\ell \in [a_i]$. Therefore,
        \[\bigcup \limits_{i \in [n], b_i = j} \left\{ Q_{b_i, \ell + p_i} \mid \ell \in [a_i] \right\} = \left\{Q_{j,\ell} \mid j = b_i, \ell \in [C]\right\}.\]
        By Claim~\ref{clm:pw-PQSU-good}, the sets on the left are pairwise disjoint. Taking cardinalities gives
        \[\sum \limits_{i \in [n], b_i = j} a_i = C.\]
        This proves the claim.
    \end{proof}
    This proves that the {\sc $k$-way Partition} instance is a \yes-instance. Thus, we prove W[1]-hardness of {\ourprob} with respect to $\Delta$ even when the pathwidth of both input graphs is a constant.
 \end{proof}

\subsection{Proof of Lemma~\ref{lem:treewidth-algo-part-1}}

\begin{proof}
Let $V = V(G)$, $E = E(G)$. Let $(T,B)$ be a nice tree decomposition of $G$ of width at most $\tw$, rooted at $r$. For a node $t \in V(T)$, let its bag be $B_t = (v_0,\ldots,v_{k-1}), k \le \tw+1$, where the vertices in the bag are given a fixed order. For a node $t \in V(T)$, we denote by $V_t$ the set of vertices that appear in the bags in the subtree of $T$ rooted at $t$, and we let $G_t = G[V_t]$.

We define a function $l : V(T) \rightarrow \mathbb{N}$ as follows. For the root $r$, we set $l(r)=0$. For any node $t \in V(T)$, we define $ l(t) = \mathrm{dist}_T(t,r)$, where $\mathrm{dist}_T(t,r)$ denotes the distance between $t$ and $r$ in the tree $T$. Let $L = \max_{t \in V(T)} l(t)$. We describe a dynamic programming algorithm over $(T,B)$ that proceeds in a bottom-up fashion with respect to $l$. All DP states corresponding to nodes $t$ with $l(t)=L$ are initialized first. Then, for $j=L-1,L-2,\ldots,0$, the DP states corresponding to nodes $t$ with $l(t)=j$ are computed assuming that the DP states for all nodes $t'$ with $l(t')>j$ have already been computed. In particular, the DP state at the root $r$ (where $l(r)=0$) is computed last.

\textbf{States.}
We maintain a dynamic programming table $\mathrm{DP}$ where a state has the
following components:
\begin{enumerate}
    \item $t$, a node of the tree decomposition $T$;
    \item a \emph{mask} $\sigma$ describing the \emph{roles} of the vertices in $B_t$;
    \item a star-count vector $\mathbf{c} = (c_2,c_3,\ldots,c_{\Delta+1}) \in [n]^\Delta$.
\end{enumerate}

\textbf{Masks.}
For each vertex $v_i \in B_t$, its \emph{role} $\sigma(v_i)$ is one of the following:
\begin{itemize}
    \item \textsc{Uncovered}: $v_i$ is not incident to any edge of the partial solution;
    \item \textsc{Centre}$(d)$ for $d \in \{2,\ldots,\Delta+1\}$: $v_i$ is the centre of a star of $d$ nodes;
    \item $\textsc{Leaf}(u)$ for some $u \in B_t$: $v_i$ is a leaf of a star whose centre is $u$; 
    \item \textsc{Leaf}($\bot$): $v_i$ is a leaf whose centre lies outside $B_t$.
\end{itemize}

A \emph{mask} for the bag $B_t$ is a tuple
$\sigma = (\sigma(v_0),\ldots,\sigma(v_{k-1})).$

\textbf{Star-count vector.}
For each $d \in \{2,\ldots,\Delta+1\}$, the value $c_d$ denotes the number of stars of size $d$ that have been formed in $G_t$.

\textbf{Interpretation of States.}
For each node $t \in V(T)$, mask $\sigma$, and star-count vector $\mathbf{c}$, the entry $\mathrm{DP}[t,\sigma,\mathbf{c}] \in \{0,1\}$ is set to $1$ if and only if there exists a subgraph $H \subseteq G_t$ such that:
\begin{itemize}
    \item $H$ is a vertex-disjoint union of stars, each of size at most $\Delta+1$;
    \item for every $d \in \{2,\ldots,\Delta+1\}$, exactly $c_d$ stars of size $d$ appear in $H$;
    \item for every vertex $v_i \in B_t$, its role in $H$ is exactly given by $\sigma(v_i)$;
    \item every edge of $H$ is an edge of $G$.
\end{itemize}
If no such subgraph exists, we set $\mathrm{DP}[t,\sigma,\mathbf{c}] = 0$.

\textbf{Initialization (Leaf node).}
If $t$ is a leaf node of $T$ with bag $B_t = \emptyset$, then $\mathrm{DP}[t,(\emptyset),(0,\ldots,0)] = 1$. For every other star-count vector $\mathbf{c'} \neq (0,\ldots,0)$,
we set
\(
\mathrm{DP}[t, \emptyset, \mathbf{c'}] = 0.
\)

\textbf{Dynamic Programming on $\mathrm{DP}$.}
The DP table is filled in a bottom-up manner along the tree $T$. For a node $t$, we assume that all DP states corresponding to the children of $t$ have already been computed, and we update $\mathrm{DP}[t,\cdot,\cdot]$ according to the type of $t$ (introduce, forget, or join).

\textbf{Introduce node.}
Let $t$ be an introduce node with child $t'$ such that $B_t = B_{t'} \cup \{u\}$, where $u \notin B_{t'}$.
Fix a mask $\sigma$ for $B_t$ and a star-count vector $\mathbf{c}$. We set $\mathrm{DP}[t,\sigma,\mathbf{c}]=1$ if and only if at least one of the following conditions holds.
\begin{enumerate}
    \item \textbf{If $u$ is uncovered.} If $\sigma(u)=\textsc{Uncovered}$, then $\mathrm{DP}[t,\sigma,\mathbf{c}]=1$ if and only if there exists a mask $\sigma'$ for $B_{t'}$ such that $\sigma'(v)=\sigma(v)$ for all $v \in B_{t'}$ and $\mathrm{DP}[t',\sigma',\mathbf{c}]=1.$
    \item \textbf{If $u$ is a leaf.} Suppose $\sigma(u)=\textsc{Leaf}(v)$ for some $v \in B_{t'}$ with $ \{u,v\} \in E(G)$. Then $\mathrm{DP}[t,\sigma,\mathbf{c}]=1$ if and only if there exists a mask $\sigma'$ and a vector $\mathbf{c}'$ such that
    \begin{itemize}
    \item $\mathrm{DP}[t',\sigma',\mathbf{c}']=1$,
    \item $\sigma'(w)=\sigma(w)$ for all $w \in B_{t'} \setminus \{v\}$,
    \item either
    \begin{itemize}
        \item $\sigma'(v)=\textsc{Uncovered}$, $c_2 = c'_2 + 1$, and $c_j=c'_j$ for all $j \neq 2$, or
        \item $\sigma'(v)=\textsc{Centre}(d)$ for some $d \le \Delta$, $c_{d+1}=c'_{d+1}+1$, $c_d=c'_d-1$, and $c_j=c'_j$ for all $j \notin \{d,d+1\}$.
    \end{itemize}
\end{itemize}
    \item \textbf{If $u$ is a centre.} Suppose $\sigma(u)=\textsc{Centre}(d)$ for some $d \in \{2,\ldots,\Delta+1\}$. Then $\mathrm{DP}[t,\sigma,\mathbf{c}]=1$ if and only if there exists a mask $\sigma'$ and a vector $\mathbf{c}'$ such that
    \begin{itemize}
    \item $\mathrm{DP}[t',\sigma',\mathbf{c}']=1$,
    \item Let $S = \{v \in B_{t'} \mid \sigma(v) = \textsc{Leaf}(u)\}$, then $S$ is of size $d-1$ with $\{u,v\} \in E(G)$ and $\sigma'(v)=\textsc{Uncovered}$ for all $v \in S$,
    \item $\sigma(w)=\sigma'(w)$ for all $w \in B_{t'} \setminus S$,
    \item $c_d = c'_d + 1$ and $c_j = c'_j$ for all $j \neq d$.
\end{itemize}
\end{enumerate}
\textbf{Forget node.}
Let $t$ be a forget node with child $t'$ such that $B_t = B_{t'} \setminus \{u\}.$ Fix a mask $\sigma$ for $B_t$ and a star-count vector $\mathbf{c}$. We set $\mathrm{DP}[t,\sigma,\mathbf{c}]=1$ if and only if there exists a mask $\sigma'$ for $B_{t'}$ such that:
\begin{itemize}
    \item $\sigma'(v)=\textsc{Leaf}(\bot)$ for all $v \in B_t$ with $\sigma(v) = \textsc{Leaf}(u)$,
    \item $\sigma'(v)=\sigma(v)$ for all $v \in B_t$ with $\sigma(v) \neq \textsc{Leaf}(u)$, and
    \item $\mathrm{DP}[t',\sigma',\mathbf{c}]=1$.
\end{itemize}
\textbf{Join node.}
Let $t$ be a join node with children $t_1$ and $t_2$ such that $B_t = B_{t_1} = B_{t_2}$. At a join node $t$, the subgraph $G_t$ is the union of two subgraphs $G_{t_1}$ and $G_{t_2}$ that intersect exactly on the bag $B_t$. A DP state at $t_1$ describes a partial star forest $H_1 \subseteq G_{t_1}$, and a DP state at $t_2$ describes a partial star forest $H_2 \subseteq G_{t_2}$. To obtain a valid solution in $G_t$, these two partial solutions must be combined into $H = H_1 \cup H_2$. The vertices of $B_t$ are the only vertices that appear in both subgraphs, and hence they are the only points where the two partial solutions can interact. Each child state assigns a role to every vertex $v \in B_t$ (\textsc{Uncovered}, \textsc{Leaf}, or \textsc{Centre} with a given size). These roles describe how $v$ participates in the partial star forest inside that subtree. However, in the final solution $H$, each vertex can have only one role in the star forest: it cannot simultaneously be a leaf and a centre, nor can it be attached to two different centres. Therefore, the role assignments coming from $t_1$ and $t_2$ must be consistent with a single global role for each $v \in B_t$. Now we discuss the vertex role compatibility.

Fix a mask $\sigma$ for $B_t$ and a star-count vector $\mathbf{c}$. We set $\mathrm{DP}[t,\sigma,\mathbf{c}] = 1$ if and only if there exist masks $\sigma_1, \sigma_2$ for $B_{t_1}, B_{t_2}$ and star-count vectors $\mathbf{c}^{(1)}, \mathbf{c}^{(2)}$ such that $\mathrm{DP}[t_1,\sigma_1,\mathbf{c}^{(1)}] = 1$ and $\mathrm{DP}[t_2,\sigma_2,\mathbf{c}^{(2)}] = 1$, and the following compatibility conditions hold.

\emph{Vertex-role compatibility.} For each $v \in B_t$, the roles $\sigma_1(v)$ and $\sigma_2(v)$ uniquely determine $\sigma(v)$ as follows.

\begin{enumerate}
    \item \textbf{If one of $\sigma_1(v)$ and $\sigma_2(v)$ is uncovered.} If a vertex is \textsc{Uncovered} in one subtree, that subtree does not use the vertex in any star. Hence, in the merged solution, the vertex’s role is determined entirely by the other subtree, since only that side contributes incident star edges for the vertex.
    \begin{enumerate}
        \item If $\sigma_1(v)$ and $\sigma_2(v)$ are both \textsc{Uncovered}, then $\sigma(v)=\textsc{Uncovered}$.
        \item If one of $\sigma_1(v),\sigma_2(v)$ is \textsc{Uncovered} and the other is $\textsc{Leaf}(u)$, then $\sigma(v)=\textsc{Leaf}(u)$. 
        \item If one of them is \textsc{Uncovered} and the other is $\textsc{Centre}(d)$, then $\sigma(v)=\textsc{Centre}(d)$. 
    \end{enumerate} 
    
     \item \textbf{If both roles $\sigma_1(v)$, $\sigma_2(v)$ are leaves.} If a vertex is a \textsc{Leaf} in both subtrees with the same vertex $u$ as centre, i.e. $\sigma_1(v) = \sigma_2(v)= \textsc{Leaf}(u)$, then it must be attached to the same centre in $B_t$; this means both partial solutions describe the same star continuing across the two subtrees, so the role remains $\sigma(v)=\textsc{Leaf}(u)$. If the two subtrees attach it to different centres, i.e. if either $\sigma_1(v) \ne \sigma_2(v)$ or $\sigma_1(v) = \sigma_2(v) = \textsc{Leaf}(\bot)$, the roles conflict and no valid star forest can result, so the combination is incompatible.

     \item \textbf{If both $\sigma_{1}(v)$ and $\sigma_{2}(v)$ are centres.} If a vertex $v \in B_t$ is a centre in both child states i.e. $\sigma_1(v)=\textsc{Centre}(d_1)$ and $\sigma_2(v)=\textsc{Centre}(d_2)$, then its incident leaves may be split between $G_{t_1}$ and $G_{t_2}$. In the combined solution, these two partial stars represent different parts of the same star and must therefore be merged into a single star. Let $S_1 = \{u \in B_t \mid \sigma_1(u) = \textsc{Leaf}(v)\}$ be the leaves of the star present in $B_{t_1}$, and $S_2 = \{u \in B_t \mid \sigma_2(u) = \textsc{Leaf}(v)\}$ be the leaves of the star present in $B_{t_2}$. $\sigma_1$ and $\sigma_2$ are compatible only if $S_1 = S_2$, i.e., the stars agree over the leaves present in the bag. In this case, let $q=|S_1|$ be the number of common leaves of $v$ already present in the bag; these are counted in both subtrees and must not be double counted. If the partial stars have sizes $d_1$ and $d_2$, the merged star has size $d = d_1 + d_2 - q - 1$, and the star-count vector is updated by removing the two partial counts and adding one count for the merged star of size $d$. If $d \le \Delta+1$, we set $\sigma(v)=\textsc{Centre}(d)$; otherwise, the combination is invalid. 
     
     \item \textbf{If one of $\sigma_{1}(v)$ or $\sigma_{2}(v)$ is a leaf and the other is a centre.} If one role is $\textsc{Leaf}(u)$ and the other is $\textsc{Centre}(d)$, the such a combination is incompatible, since a vertex cannot simultaneously act as a leaf and a centre in the same star forest.

\end{enumerate}

\emph{Star-count update.}
For each $d \in \{2,\ldots,\Delta+1\}$, let $a_d$ be the number of vertices $v\in B_t$ with $\sigma_1(v)=\textsc{Centre}(d)$ and $b_d$ the number with $\sigma_2(v)=\textsc{Centre}(d)$. Let $m_d$ denote the number of vertices $v\in B_t$ that are centres in both subtrees and whose two partial stars combine into a single star of size $d$ in the merged solution. A simple application of the inclusion-exclusion principle gives us that the resulting star-count vector satisfy $c_d = c^{(1)}_d + c^{(2)}_d - a_d - b_d + m_d$. We show how to compute all such $\mathbf{c}$ efficiently.

We pick $\sigma_1$, $\sigma_2$ which are compatible as discussed in vertex-role compatibility. We set \(\mathcal{C}_1= \{\mathbf{c} \mid \mathrm{DP}[t_1, \sigma_1, \mathbf{c}] = 1 \}\) and similarly, \(\mathcal{C}_2= \{\mathbf{c'} \mid \mathrm{DP}[t_1, \sigma_1, \mathbf{c'}] = 1 \}\). We can compute $\mathcal{C} = \{ \mathbf{c} + \mathbf{c'} \mid \mathbf{c} \in \mathcal{C}_1, \mathbf{c'} \in \mathcal{C}_2\}$ using Lemma~\ref{lem:sumset}, in time $\OO(\Delta n^\Delta\log n)$. Let $\mathbf{a} = (a_2,a_3,\ldots,a_{\Delta+1})$, $\mathbf{b} = (b_2,b_3,\ldots,b_{\Delta+1})$ and $\mathbf{m} = (m_2,m_3,\ldots,m_{\Delta+1})$. We define $\mathcal{C}^* = \{\mathbf{c}-\mathbf{a}-\mathbf{b}+\mathbf{m} \mid \mathbf{c} \in \mathcal{C}\}$. We set $\mathrm{DP}[t,\sigma,\mathbf{c}]$ to $1$ if $\mathbf{c} \in \mathcal{C}^*$, and $0$ otherwise.

Let $r$ be the root of the tree decomposition. By construction of a nice
tree decomposition, the bag $B_r$ is empty. Hence any state at $r$ has
the form $\mathrm{DP}[r,(\emptyset),\mathbf{c}]$.

Since $B_r=\emptyset$, all vertices of $H$ lie strictly inside the subtree. Therefore $H$ is a star forest subgraph in $G$. 
The algorithm outputs the following star forests: For every $\mathbf{c} = (c_2, \ldots, c_{\Delta+1}) \in [n]^\Delta$ such that $\mathrm{DP}[r, (\emptyset), \mathbf{c}] = 1$, output the star forest containing $c_i$ copies of the star of size $i$ for every $i \in \{2, 3, \ldots, \Delta + 1\}$. 

\textbf{Proof of correctness.} 
We prove by induction over the tree decomposition that the dynamic
programming table satisfies the intended semantics.

\begin{claim}
For every node $t \in V(T)$, mask $\sigma$ on $B_t$, and star-count vector $\mathbf{c}$, we have $\mathrm{DP}[t,\sigma,\mathbf{c}] = 1$ if and only if there exists a subgraph $H \subseteq G_t$ such that 
\begin{enumerate}
\item $H$ is a vertex-disjoint union of stars, each of size at most $\Delta+1$,
\item for every $d \in \{2,\ldots,\Delta+1\}$, exactly $c_d$ stars of size $d$ appear in $H$,
\item for every vertex $v \in B_t$, its role in $H$ is exactly $\sigma(v)$,
\item every edge of $H$ belongs to $G$.
\end{enumerate}
\end{claim}

\begin{proof}
We prove the claim by induction on $l(t)$, that is, in the bottom-up order of the tree decomposition.
 
\textbf{Forward direction.}
\textbf{Base case (leaf node).}
At a leaf $t$, both $B_t$ and $G_t$ are empty. The only nonzero entry is $\mathrm{DP}[t,(\emptyset),(0,\ldots,0)] = 1$. The empty graph is the only subgraph of $G_t$ and satisfies the claim with no stars and no bag vertices.

\textbf{Inductive step.}
Assume the claim holds for all children of a node $t$. We prove it for $t$ according to its type.

\textbf{Introduce node.}
Suppose $t$ introduces a vertex $u$, i.e., $B_t = B_{t'} \cup \{u\}$ for its child $t'$. Any star forest $H \subseteq G_t$ restricts to a star forest $H' = H[V_{t'}] \subseteq G_{t'}$. By the induction hypothesis, this corresponds to a state $\mathrm{DP}[t',\sigma',\mathbf{c}'] = 1$.

The transitions for introduce nodes consider exactly the three possible roles of $u$ in $H$:
\begin{enumerate}
    \item If $u$ is \textsc{Uncovered}, then $H = H'$ and the DP keeps the same star counts and roles on $B_{t'}$.
    \item If $u$ is a \textsc{Leaf}$(v)$, then $u$ is attached to some centre $v$ with $\{u,v\} \in E(G)$. The DP transitions either create a new star of size $2$ (if $v$ was uncovered) or extend an existing star centred at $v$, updating the corresponding entry of $\mathbf{c}$.
    \item  If $u$ is a \textsc{Centre}, then $u$ becomes the centre of a star with leaves among its neighbours in $B_{t'}$, and the DP verifies that these vertices were previously uncovered and updates the star counts.
\end{enumerate}

Thus every DP transition constructs a valid extension of a star forest from $G_{t'}$ to $G_t$. Hence the claim holds for introduce nodes.

\textbf{Forget node.}
Suppose $t$ forgets a vertex $u$, i.e., $B_t = B_{t'} \setminus \{u\}$. Any star forest $H \subseteq G_t$ is also a star forest in $G_{t'}$ with the same star counts. The role of $u$ in $H$ is no longer relevant once $u \notin B_t$. The DP transition performs an \textsc{or} operation over all possible roles of $u$ in the child state, preserving the star-count vector. By the induction hypothesis, this exactly captures all valid star forests in $G_t$. Thus the claim holds for forget nodes. 

\textbf{Join node.}
Let $t$ be a join node with children $t_1,t_2$ and $B_t = B_{t_1} = B_{t_2}$. Any star forest $H \subseteq G_t$ decomposes into $H_1 = H \cap G_{t_1}$ and $H_2 = H \cap G_{t_2}$, which agree on $B_t$. By the induction hypothesis, $H_1$ and $H_2$ correspond to states $\mathrm{DP}[t_1,\sigma_1,\mathbf{c}^{(1)}]=1$ and $\mathrm{DP}[t_2,\sigma_2,\mathbf{c}^{(2)}]=1$. The vertex-role compatibility conditions ensure that the roles of vertices in the bag are consistent between $H_1$ and $H_2$. In particular, stars centred at bag vertices that appear in both subtrees are merged correctly, and invalid combinations (e.g., a vertex being both leaf and centre) are discarded. The star-count update subtracts stars counted in both subtrees whose centres lie in the bag and adds back merged stars with updated sizes. Thus the vector $\mathbf{c}$ records exactly the number of stars in $H$.

Hence the claim holds for join nodes. By induction, the claim holds for the root $r$. Therefore, the DP correctly decides whether a star forest with the required size exists.

\textbf{Backward direction.}
Assume a star forest $H \subseteq G_t$ satisfies Properties (1)--(4). We show $\mathrm{DP}[t,\sigma,\mathbf{c}] = 1$.

\textbf{Base case (leaf node).} If $t$ is a leaf, $H$ must be empty, so the DP initialization sets $\mathrm{DP}[t,(\emptyset),(0,\ldots,0)] = 1$.

\textbf{Inductive step.} Assume the claim holds for all children of a node $t$. We prove it for $t$ according to its type.

\textbf{Introduce node.} Let $t$ introduce $u$. Let $H' = H[V_{t'}]$. By induction, $H'$ corresponds to $\mathrm{DP}[t',\sigma',\mathbf{c}']=1$. The role of $u$ in $H$ determines which introduce transition applies, updating roles and counts to $\sigma,\mathbf{c}$. Hence $\mathrm{DP}[t,\sigma,\mathbf{c}] = 1$.

\textbf{Forget node.} If $t$ forgets $u$, then $H$ is also a star forest in $G_{t'}$. By induction, it corresponds to $\mathrm{DP}[t',\sigma',\mathbf{c}] = 1$. The forget transition projects $\sigma'$ to $B_t$, so $\mathrm{DP}[t,\sigma,\mathbf{c}] = 1$.

\textbf{Join node.} Let $H_i = H \cap G_{t_i}$ for $i=1,2$. By induction, these correspond to states $\mathrm{DP}[t_i,\sigma_i,\mathbf{c}^{(i)}]=1$. Since both originate from the same $H$, their roles on $B_t$ are compatible. The join transition merges them and sets $\mathrm{DP}[t,\sigma,\mathbf{c}] = 1$.

Thus both directions hold, completing the proof.
\end{proof}

\textbf{Running-time analysis.} 
Let $k = \tw+1$ denote the maximum bag size. For each vertex $v$ in a bag, its role can be \textsc{Uncovered}, \textsc{Centre}$(d)$ for some $d \in \{2,\ldots,\Delta+1\}$, or \textsc{Leaf}$(u)$ for $u \in B_t \setminus \{v\}$ or $\textsc{Leaf}(\bot)$. Hence each vertex has at most $1+\Delta+(k+1) = \Delta + \tw + 3$ possible roles, and the number of masks per bag is at most $(\Delta + \tw + 3)^{\tw+1}$. The star-count vector has $\Delta$ coordinates, each bounded by $n$, so the number of possible vectors is $\OO(n^\Delta)$. Therefore, the number of DP states per bag is $(\Delta + \tw + 3)^{\tw+1} \cdot n^\Delta.$ Introduce and forget nodes can be processed in time polynomial in ($\tw+\Delta$) per state. At a join node, in the worst case every state of one child must be checked against every state of the other child, leading to a time of $(\Delta + \tw + 3)^{2(\tw+1)} \cdot \Delta \cdot n^{\Delta} \log n$ for that node. Since a nice tree decomposition has $\OO(n \cdot \tw)$ nodes, the total running time of the algorithm is $\OO\big((\Delta + \tw + 3)^{2(\tw+1)} \cdot \tw \cdot \Delta \cdot n^{\Delta+1} \cdot \log n\big)$.
\end{proof}

\end{document}